\documentclass[11pt, a4paper]{article}

\usepackage{amsmath, amsfonts, amssymb, amsthm}
\usepackage{graphicx}
\usepackage{epsfig}
\usepackage{float}
\usepackage{color}
\usepackage[top=1in, bottom=1in, left=1in, right=1in]{geometry}
\usepackage{url}
\usepackage{comment}
\usepackage[linesnumbered,ruled,vlined]{algorithm2e}
\usepackage{mathtools}
\usepackage{booktabs}
\SetCommentSty{mycommfont}

\SetKwInput{KwInput}{Input}
\SetKwInput{KwOutput}{Output}
\SetKwInput{KwFunction}{Function}
\DeclareMathOperator{\wt}{wt}

\DeclareMathOperator{\rank}{rank}

\DeclareMathOperator{\Circ}{Circ}
\DeclareMathOperator{\RevCirc}{RevCirc}

\numberwithin{figure}{section}

\newtheorem{theorem}{Theorem}[section]
\newtheorem{lemma}[theorem]{Lemma}
\newtheorem{corollary}[theorem]{Corollary}
\newtheorem{proposition}[theorem]{Proposition}
\newtheorem{example}[theorem]{Example}
\newtheorem{remark}[theorem]{Remark}
\newtheorem{definition}[theorem]{Definition}

\begin{document}

\title{Quantum Codes from Group Codes}

\author{Tushar Bag$^{1,2}$,  Daniel Panario$^3$}
\date{
\small{
1. Department of Mathematics, SRM University-AP, Amaravati 522240, Andhra Pradesh, India \\
2. Inria, ENS de Lyon, Lyon 1, LIP, 69342, Lyon cedex 07, France\\
3. School of Mathematics and Statistics, Carleton University, Ottawa, Ontario K1S 5B6, Canada.
}
\today}

\begingroup
\renewcommand\thefootnote{}
\footnotetext{
Email: tusharbag2011@gmail.com (T. Bag) [corresponding author],
daniel@math.carleton.ca (D. Panario)
}
\endgroup

\maketitle

\begin{abstract}
We study linear codes and quantum error-correcting codes (QECCs) constructed from group rings over finite fields. Using the algebraic structure of group rings, we give a single framework for codes over several group structures, including cyclic, dihedral, direct-product, and semidirect-product groups. We establish necessary and sufficient conditions for these group codes to be self-orthogonal under the Euclidean, Hermitian, and symplectic inner products. We show that non-isomorphic groups of the same order can generate inequivalent codes with distinct parameters, and we support this with explicit computational comparisons. Using these structural results, we give explicit constructions of quantum codes and provide new examples that match or improve upon the best known parameters.
In particular, we describe explicit block-matrix forms of the generating matrices for dihedral and direct-product groups, and we use a Kronecker-product construction to obtain an infinite family of self-orthogonal group codes together with the corresponding QECCs.
\end{abstract}

\noindent\text{\bf Keywords:} \small{Group rings, quantum error-correcting codes (QECCs).}\\
\text{\bf Mathematics Subject Classification (2010):} \small{94B05, 94B15, 94B60.}\\

\normalsize

\section{Introduction}

Group codes were introduced by Berman \cite{B67a, B67b}, with an initial focus on abelian and semisimple abelian codes. Significant theoretical advances were established in subsequent works \cite{BDS09, HHH18, HH03, HH09, M69, M70, W06}. In particular, Hurley \cite{H06} constructed an explicit isomorphism between a group ring and a matrix ring, which simplifies the encoding and analysis of group codes. This matrix representation is the main algebraic tool used in our study. More recently, group codes employing this matrix approach have been explored by Dougherty et al.\ \cite{D23, DGTT18, DKSU23} and by Yu and Zhu \cite{YZ23, YZ24}. Borello and Willems \cite{BW20} proved that group codes over finite fields are asymptotically good in every characteristic. Borello's lecture notes \cite{Bor21} provide an overview of the algebraic structure of group codes and their connections with classical linear codes. Related work by Borello and Jamous \cite{BJ18} examines the automorphism groups of binary linear codes, while group codes have also proven useful for constructing good linear codes \cite{H24, JLLX10}.

\vskip 3pt
Quantum error-correcting codes (QECCs) protect quantum information against decoherence and quantum noise. Quantum computers offer significant computational potential, but they remain susceptible to errors caused by environmental interactions, and QECCs are a central tool for protecting quantum information in both computing and communication. The foundational concepts of QECCs were introduced by Shor, Steane, and Calderbank \cite{CS96, Sh95, St96}, and the Calderbank--Shor--Steane (CSS) framework \cite{CRSS98} continues to play a fundamental role in quantum coding theory. Techniques for constructing non-binary quantum codes from classical self-orthogonal codes have been developed in subsequent works \cite{AKS07, AK01, KKKS06}.

\vskip 3pt
Our motivation for studying linear codes through group rings, and for using them to construct quantum codes, is the following. For a fixed length, group rings provide several different generating matrices for a code. For example, to construct a code of length $12$ over $\mathbb{F}_q$, one can use the group ring $\mathbb{F}_q[G]$ for any group $G$ of order $12$. Up to isomorphism there are five such groups: the cyclic group $C_{12}$, the dihedral group $D_6$, the direct product $C_2 \times C_6$, the alternating group $A_4$, and the dicyclic group $C_3 \rtimes C_4$. Here $C_n$ is the cyclic group of order $n$, $D_n$ is the dihedral group of order $2n$, $A_4$ is the alternating group of order $12$, $\times$ denotes the direct product of groups, and $\rtimes$ denotes the semidirect product. By exploring these distinct, non-isomorphic structures we gain flexibility in code design and can obtain larger minimum distances than those available from cyclic codes of the same length, as confirmed by the comparisons in Section~\ref{sec:comparison}.

\vskip 3pt
A further advantage of the group ring formalism is that it extends naturally to non-abelian groups, such as dihedral groups, direct products of cyclic and dihedral groups, and semidirect products, all within one matrix description and without recourse to non-commutative quotient rings. This uniformity is what lets the duality criteria and the quantum constructions of Section~\ref{sec:qecc} be stated once and applied across all of these families.

\vskip 3pt
This paper is organized as follows. Section~\ref{sec:prelim} establishes the preliminaries on group rings and their correspondence with linear codes via Hurley's matrix homomorphism. In Section~\ref{sec:gencodes}, we describe the block-matrix generation of group codes derived from various group structures, including $D_n$, $C_l \times C_m$, and $C_l \times D_m$. Section~\ref{sec:ordering} studies the effect of the ordering of group elements on the resulting generating matrices. In Section~\ref{sec:distinct}, we prove that non-isomorphic groups of the same order can generate inequivalent codes; this is supported by computational comparisons in Section~\ref{sec:comparison}. Section~\ref{sec:qecc} establishes necessary and sufficient conditions for group codes to be self-orthogonal under the Euclidean, Hermitian, and symplectic inner products, and gives explicit constructions of quantum codes that match or improve upon known parameters. Section~\ref{sec:conclusion} concludes the paper.
\section{Preliminaries}\label{sec:prelim}
Let $G$ be a finite group and $\mathbb{F}_q$ the finite field with $q = p^r$ elements, where $p$ is prime. The \textit{group ring} $\mathbb{F}_q[G]$ consists of all formal $\mathbb{F}_q$-linear combinations of group elements,
\[
a = \sum_{g \in G} \alpha_g\, g, \quad \alpha_g \in \mathbb{F}_q,
\]
with only finitely many nonzero coefficients. Writing $G = \{g_1, \ldots, g_n\}$ for a fixed ordering, addition in $\mathbb{F}_q[G]$ is componentwise, and multiplication extends bilinearly from the group law:
\[
\Bigl(\sum_{g \in G} \alpha_g\, g\Bigr)\Bigl(\sum_{h \in G} \beta_h\, h\Bigr)
= \sum_{g' \in G} \gamma_{g'}\, g', \qquad \gamma_{g'} = \sum_{gh = g'} \alpha_g \beta_h.
\]
These operations make $\mathbb{F}_q[G]$ a unital ring---an $n$-dimensional $\mathbb{F}_q$-algebra---with $|\mathbb{F}_q[G]| = q^n$ elements and unity $1_G$ (the identity of $G$ viewed as a ring element). Scalar multiplication is $\lambda \cdot a := \sum_{g} (\lambda\alpha_g)g$ for $\lambda \in \mathbb{F}_q$.

\medskip
We identify $\mathbb{F}_q[G]$ with $\mathbb{F}_q^n$ as $\mathbb{F}_q$-vector spaces via the bijection
\[
\phi \colon \mathbb{F}_q^n \to \mathbb{F}_q[G], \quad (\alpha_1, \ldots, \alpha_n) \mapsto \sum_{i=1}^n \alpha_i g_i.
\]
Under $\phi$, the standard Hamming weight on $\mathbb{F}_q^n$ corresponds to the \textit{weight} of a group ring element: for $a = \sum_{i=1}^n \alpha_i g_i$,
\[
\wt(a) = |\operatorname{supp}(a)|, \qquad \operatorname{supp}(a) := \{g_i \in G : \alpha_i \neq 0\},
\]
and the Hamming distance between $a, b \in \mathbb{F}_q[G]$ is $d(a,b) = \wt(a - b)$.

\medskip
We now describe the codes we study.

\begin{definition}\label{def:groupcode_prelim}
A linear code $C \subseteq \mathbb{F}_q^n$ is a group code if $\phi(C)$ is a left ideal of $\mathbb{F}_q[G]$; that is, $g \cdot a \in \phi(C)$ for every $a \in \phi(C)$ and $g \in G$.
\end{definition}

\noindent The connection to classical cyclic codes is immediate: when $G = C_n$ is the cyclic group of order $n$, the natural isomorphism $\mathbb{F}_q[C_n] \cong \mathbb{F}_q[x]/\langle x^n - 1\rangle$ recovers the usual description of cyclic codes as ideals in the quotient polynomial ring. Group codes over non-cyclic groups thus directly generalize this framework.

\medskip
To study these codes computationally, we use the matrix embedding introduced by Hurley~\cite{H06}.

\begin{theorem}[\cite{H06}]\label{thm:hurley}
Let $G = \{g_1, \ldots, g_n\}$ be a finite group and $a = \sum_{k=1}^n \alpha_{g_k} g_k \in \mathbb{F}_q[G]$. Define $\sigma(a) \in M_n(\mathbb{F}_q)$ by
\[
\bigl(\sigma(a)\bigr)_{ij} = \alpha_{g_i^{-1} g_j}.
\]
Then $\sigma \colon \mathbb{F}_q[G] \to M_n(\mathbb{F}_q)$ is an injective $\mathbb{F}_q$-algebra homomorphism.
\end{theorem}

Finally, we transport the standard inner product to the group ring setting. For $a = \sum_i \alpha_i g_i$ and $b = \sum_i \beta_i g_i$ in $\mathbb{F}_q[G]$, their \textit{Euclidean inner product} is
\[
\langle a, b \rangle_e = \sum_{i=1}^n \alpha_i \beta_i .
\]
The \textit{Euclidean dual} of a code $C \subseteq \mathbb{F}_q[G]$ is
\[
C^{\perp_e} = \{a \in \mathbb{F}_q[G] : \langle a, b\rangle_e = 0 \text{ for all } b \in C\}.
\]
We say $C$ is \textit{self-orthogonal} if $C \subseteq C^{\perp_e}$, and \textit{dual-containing} if $C^{\perp_e} \subseteq C$.

\section{Group Codes and Their Generating Matrices}\label{sec:gencodes}
A generator matrix of a linear code may be presented in two ways. In standard (row-reduced) form its rows are linearly independent and their number equals the dimension of the code. For codes carrying additional algebraic symmetry it is often more convenient to work with a square spanning matrix whose row space is the code; such a matrix may have linearly dependent rows. Throughout this paper we call this square spanning matrix the \emph{generating matrix} of the code, and a standard generator matrix is recovered by selecting any maximal linearly independent subset of its rows. For group codes the natural generating matrix is the one attached to an element of the underlying group ring via Theorem~\ref{thm:hurley}.

We use two square arrays throughout our constructions: the \emph{circulant} and \emph{reverse circulant} matrices,
\[
\Circ(a_1,a_2,\dots,a_n) = \begin{pmatrix}
a_1  & a_2  & \cdots & a_n  \\
a_n  & a_1  & \cdots & a_{n-1}  \\
\vdots & \vdots & \ddots & \vdots \\
a_2  & a_3  & \cdots & a_{1}  \\
\end{pmatrix},~~
\RevCirc(a_1,a_2,\dots,a_n) = \begin{pmatrix}
a_1  & a_2  & \cdots & a_n  \\
a_2  & a_3  & \cdots & a_{1}  \\
\vdots & \vdots & \ddots & \vdots \\
a_n  & a_1  & \cdots & a_{n-1}  \\
\end{pmatrix}.
\]
Thus each row of a circulant is the right cyclic shift of the previous one, while each row of a reverse circulant is the left cyclic shift.

We recall two elementary properties of circulant matrices.

\begin{lemma}\label{lem:circ-algebra}
The set of $n\times n$ circulant matrices over $\mathbb{F}_q$ is a commutative $\mathbb{F}_q$-algebra, and the map
\[
\Circ(a_1,a_2,\dots,a_n)\longmapsto a_1+a_2 t+\cdots+a_n t^{\,n-1}
\]
is an $\mathbb{F}_q$-algebra isomorphism onto $R_n:=\mathbb{F}_q[t]/(t^n-1)$. Consequently, writing $a(t)=a_1+a_2t+\cdots+a_nt^{\,n-1}$,
\[
\Circ(a_1,\dots,a_n)\ \text{is nonsingular if and only if}  \gcd\!\big(a(t),t^n-1\big)=1,
\]
and in general $\operatorname{rank}\Circ(a_1,\dots,a_n)=n-\deg\gcd\!\big(a(t),t^n-1\big)$.
\end{lemma}

\begin{proof}
The displayed map is precisely the matrix $\sigma$ of Theorem~\ref{thm:hurley} applied to the cyclic group $C_n=\langle t\rangle$, identified with $R_n$; hence it is an injective unital $\mathbb{F}_q$-algebra homomorphism, and commutativity is inherited from $R_n$. Since $R_n$ is a principal ideal ring, the ideal $R_n\,a(t)$ equals $\big(\gcd(a(t),t^n-1)\big)$, whose $\mathbb{F}_q$-dimension is $n-\deg\gcd(a(t),t^n-1)$; this dimension is the rank of $\Circ(a_1,\dots,a_n)$.
\end{proof}

\begin{lemma}\label{lem:revcirc}
Let $J$ denote the $n\times n$ exchange (anti-diagonal) matrix. Then
\[
\RevCirc(a_1,a_2,\dots,a_n)=\Circ(a_n,a_{n-1},\dots,a_1)\,J .
\]
In particular every reverse circulant matrix is symmetric, and
\[
\operatorname{rank}\RevCirc(a_1,\dots,a_n)=\operatorname{rank}\Circ(a_n,a_{n-1},\dots,a_1).
\]
\end{lemma}

\begin{proof}
The $(i,j)$-entry of $\RevCirc(a_1,\dots,a_n)$ is $a_{1+((i+j-2)\bmod n)}$, which depends only on $i+j$; hence the matrix is symmetric. A direct computation of the $(i,j)$-entry of $\Circ(a_n,\dots,a_1)J$ gives the same value, which proves the identity. As $J$ is a permutation matrix it is invertible, so multiplication by $J$ preserves rank.
\end{proof}

\begin{definition}\label{def:groupcode}
For $c\in\mathbb{F}_q[G]$, the \emph{group code} $\mathcal{C}(c)$ is the $\mathbb{F}_q$-row space of $\sigma(c)$; it is the image under $\sigma$ of the left ideal $\mathbb{F}_q[G]\,c$. Its dimension is $\dim_{\mathbb{F}_q}\mathcal{C}(c)=\operatorname{rank}\sigma(c)$.
\end{definition}

A useful consequence of Theorem~\ref{thm:hurley} is that $\sigma(c)$ is nonsingular if and only if $c$ is a unit of $\mathbb{F}_q[G]$, in which case $\mathcal{C}(c)=\mathbb{F}_q^{\,n}$; the codes of interest therefore arise from the zero divisors of the group ring.

\subsection{Self-Dual, Self-Orthogonal, and LCD Group Codes}\label{subsec:duality}
We equip $\mathbb{F}_q^{\,n}$ with the Euclidean inner product $\langle u,v\rangle=\sum_{i=1}^n u_iv_i$ and write $\mathcal{C}^\perp$ for the dual code. The group ring carries the \emph{canonical involution}
\[
c^*=\Big(\sum_{g\in G}\alpha_g\,g\Big)^{\!*}=\sum_{g\in G}\alpha_g\,g^{-1},
\]
an $\mathbb{F}_q$-linear anti-automorphism satisfying $(cd)^*=d^*c^*$ and $c^{**}=c$. The following lemma links the matrix transpose to this involution and underlies the duality results below.

\begin{lemma}\label{lem:involution}
For every $c\in\mathbb{F}_q[G]$ we have $\sigma(c^*)=\sigma(c)^{\mathsf T}$.
\end{lemma}

\begin{proof}
Write $c=\sum_{g\in G}\alpha_g g$, so that $c^*=\sum_{g\in G}\alpha_g g^{-1}$. The $(r,s)$-entry of $\sigma(c^*)$ is the coefficient of $g_r^{-1}g_s$ in $c^*$. Since the coefficient of $h$ in $c^*$ is $\alpha_{h^{-1}}$,
\[
(\sigma(c^*))_{r,s}=\alpha_{(g_r^{-1}g_s)^{-1}}=\alpha_{g_s^{-1}g_r}.
\]
By the definition of $\sigma(c)$,
\[
(\sigma(c))_{s,r}=\alpha_{g_s^{-1}g_r}.
\]
Hence $(\sigma(c^*))_{r,s}=(\sigma(c))_{s,r}=(\sigma(c)^{\mathsf T})_{r,s}$ for every pair $(r,s)$, so $\sigma(c^*)=\sigma(c)^{\mathsf T}$.
\end{proof}

\begin{proposition}\label{prop:dual-is-group}
For every $c\in\mathbb{F}_q[G]$,
\[
\mathcal{C}(c)^\perp=\{\,v\in\mathbb{F}_q[G]:\ v\,c^*=0\,\},
\]
which is a left ideal of $\mathbb{F}_q[G]$. In particular the Euclidean dual of a group code is again a group code, and $\dim\mathcal{C}(c)^\perp=n-\operatorname{rank}\sigma(c)$.
\end{proposition}

\begin{proof}
For $h\in G$ one has $\langle uh,v\rangle=\langle u,vh^{-1}\rangle$, and extending bilinearly gives $\langle ux,v\rangle=\langle u,vx^*\rangle$ for all $u,v,x\in\mathbb{F}_q[G]$. Since $\mathcal{C}(c)=\{uc:u\in\mathbb{F}_q[G]\}$, a vector $v$ lies in $\mathcal{C}(c)^\perp$ if and only if $\langle uc,v\rangle=\langle u,vc^*\rangle=0$ for all $u$, that is, if and only if $vc^*=0$ (the form is nondegenerate). This set is closed under left multiplication, hence a left ideal. The dimension is the standard duality identity.
\end{proof}

We now record the identity that underlies the duality results of this subsection. Combining Theorem~\ref{thm:hurley} with Lemma~\ref{lem:involution},
\begin{equation}\label{eq:gram}
\sigma(c)\,\sigma(c)^{\mathsf T}=\sigma(c)\,\sigma(c^*)=\sigma(cc^*),
\end{equation}
so the matrix of inner products of the rows of $\sigma(c)$ is exactly $\sigma(cc^*)$. The next two results read off self-duality and the LCD property directly from the ring element $cc^*$.

\begin{theorem}\label{thm:selfdual}
Let $n=|G|$ and $c\in\mathbb{F}_q[G]$.
\begin{enumerate}
\item[\textup{(i)}] $\mathcal{C}(c)$ is self-orthogonal if and only if $cc^*=0$.
\item[\textup{(ii)}] $\mathcal{C}(c)$ is self-dual if and only if $cc^*=0$ and $\operatorname{rank}\sigma(c)=n/2$.
\end{enumerate}
\end{theorem}

\begin{proof}
(i) The containment $\mathcal{C}(c)\subseteq\mathcal{C}(c)^\perp$ holds if and only if  all pairwise inner products of the rows of $\sigma(c)$ vanish, i.e.\ $\sigma(c)\sigma(c)^{\mathsf T}=0$. By \eqref{eq:gram} this is $\sigma(cc^*)=0$, equivalently $cc^*=0$ since $\sigma$ is injective. (ii) A self-orthogonal code is self-dual precisely when $\dim\mathcal{C}(c)=\operatorname{rank}\sigma(c)=n/2$.
\end{proof}

\begin{theorem}\label{thm:lcd}
A group code $\mathcal{C}(c)$ is an LCD code, i.e.\ $\mathcal{C}(c)\cap\mathcal{C}(c)^\perp=\{0\}$, if and only if
\[
\operatorname{rank}\sigma(cc^*)=\operatorname{rank}\sigma(c).
\]
\end{theorem}

\begin{proof}
The radical of the inner product restricted to $\mathcal{C} = \mathcal{C}(c)$ is exactly $\mathcal{C} \cap \mathcal{C}^\perp$, so
\[
\dim(\mathcal{C} \cap \mathcal{C}^\perp) = \dim\mathcal{C} - \operatorname{rank}\Gamma,
\]
where $\Gamma$ is the Gram matrix of any spanning set of $\mathcal{C}$. We take the rows of the generating matrix $\sigma(c)$ as the spanning set; by \eqref{eq:gram} their Gram matrix is
\[
\Gamma = \sigma(c)\sigma(c)^{\mathsf T} = \sigma(cc^*).
\]
Since $\dim\mathcal{C} = \operatorname{rank}\sigma(c)$, we obtain
\[
\dim(\mathcal{C} \cap \mathcal{C}^\perp) = \operatorname{rank}\sigma(c) - \operatorname{rank}\sigma(cc^*).
\]
Hence $\mathcal{C} \cap \mathcal{C}^\perp = \{0\}$ if and only if $\operatorname{rank}\sigma(cc^*) = \operatorname{rank}\sigma(c)$.
\end{proof}

\begin{corollary}\label{cor:lcd-idem}
Let $\mathcal{C}(c)=\mathbb{F}_q[G]\,e$ for an idempotent $e$ with $e=e^*$.
Then $\mathcal{C}(c)$ is LCD, and $\mathbb{F}_q[G]=\mathbb{F}_q[G]e\oplus
\mathbb{F}_q[G](1-e)$ is an orthogonal direct sum.
\end{corollary}

\begin{proof}
Let $c = e$. Since $e^* = e$ and $e^2 = e$, we have $cc^* = ee^* = e^2 = e$, so $\operatorname{rank}\sigma(cc^*) = \operatorname{rank}\sigma(c)$; by Theorem~\ref{thm:lcd}, $\mathcal{C}(c)$ is LCD. For the orthogonal direct sum, Proposition~\ref{prop:dual-is-group} gives
\[
\mathcal{C}(c)^\perp = \{v \in \mathbb{F}_q[G] : ve^* = 0\} = \{v \in \mathbb{F}_q[G] : ve = 0\} = \mathbb{F}_q[G](1-e),
\]
the last equality because $1-e$ is the complementary idempotent. Finally $1 = e + (1-e)$ yields
\[
\mathbb{F}_q[G] = \mathbb{F}_q[G]e \oplus \mathbb{F}_q[G](1-e),
\]
and the second summand is exactly $\mathcal{C}(c)^\perp$, so the sum is orthogonal.
\end{proof}

\subsection{The Group Ring $\mathbb{F}_q[D_n]$}
Consider the dihedral group $D_n=\langle a,b\mid a^n=b^2=e,\ bab^{-1}=a^{-1}\rangle$ of order $2n$. Since $ba=a^{-1}b$, and hence $ab=ba^{-1}$, the elements of $D_n$ admit the two natural listings
\begin{enumerate}
    \item $D_n=\{e,a,a^2,\dots,a^{n-1},\,b,ba,ba^2,\dots,ba^{n-1}\}$,
    \item $D_n=\{e,a,a^2,\dots,a^{n-1},\,b,ab,a^2b,\dots,a^{n-1}b\}$.
\end{enumerate}

\begin{remark}{\em
The two listings contain the same elements in different positions: in the first, $ba$ occupies position $n+2$, while in the second the element in that position is $a^{n-1}b=a^{-1}b=ba$.
}\end{remark}

The dihedral structure is governed by the following dichotomy, which explains the simultaneous appearance of circulant and reverse circulant blocks.

\begin{lemma}\label{lem:dichotomy}
In $\mathbb{F}_q[D_n]$, left multiplication by $a$ stabilizes the rotation set $\{a^i\}$ and the reflection set $\{ba^i\}$, acting as $a^i\mapsto a^{i+1}$ on the former and as $ba^i\mapsto ba^{i-1}$ on the latter (indices read modulo $n$).
\end{lemma}

\begin{proof}
Clearly $a\cdot a^i=a^{i+1}$. Using $ab=ba^{-1}$ we obtain $a\cdot ba^i=(ab)a^i=ba^{-1}a^i=ba^{i-1}$.
\end{proof}

\subsubsection{Form 1.}
With the listing $\{e,a,\dots,a^{n-1},b,ba,\dots,ba^{n-1}\}$, every element of $\mathbb{F}_q[D_n]$ has the form
\[
c=\alpha_0+\alpha_1a+\cdots+\alpha_{n-1}a^{n-1}
+\beta_0b+\beta_1ba+\cdots+\beta_{n-1}ba^{n-1},\qquad \alpha_i,\beta_i\in\mathbb{F}_q.
\]
The first row of the generating matrix is the coefficient vector of $c$; the second is that of $ac$. Since $a^n=e$ and, by Lemma~\ref{lem:dichotomy}, $a\cdot ba^i=ba^{i-1}$,
\[
ac=\alpha_{n-1}+\alpha_0a+\cdots+\alpha_{n-2}a^{n-1}
+\beta_1b+\beta_2ba+\cdots+\beta_{n-1}ba^{n-2}+\beta_0ba^{n-1}.
\]
Continuing in this way, the rows indexed by $e,a,\dots,a^{n-1}$ produce a circulant block in the rotation coordinates and a reverse circulant block in the reflection coordinates, while the rows indexed by $b,ba,\dots,ba^{n-1}$ produce the same two blocks interchanged (using $b^2=e$). Hence
\[
\mathcal{G}_1=\begin{pmatrix}
\alpha_0 & \alpha_1 & \cdots & \alpha_{n-1} & \beta_0 & \beta_1 & \cdots & \beta_{n-1} \\
\alpha_{n-1} & \alpha_0 & \cdots & \alpha_{n-2} & \beta_1 & \beta_2 & \cdots & \beta_0 \\
\vdots & \vdots & \ddots & \vdots & \vdots & \vdots & \ddots & \vdots \\
\alpha_1 & \alpha_2 & \cdots & \alpha_0 & \beta_{n-1} & \beta_0 & \cdots & \beta_{n-2} \\
\beta_0 & \beta_1 & \cdots & \beta_{n-1} & \alpha_0 & \alpha_1 & \cdots & \alpha_{n-1} \\
\beta_1 & \beta_2 & \cdots & \beta_0 & \alpha_{n-1} & \alpha_0 & \cdots & \alpha_{n-2} \\
\vdots & \vdots & \ddots & \vdots & \vdots & \vdots & \ddots & \vdots \\
\beta_{n-1} & \beta_0 & \cdots & \beta_{n-2} & \alpha_1 & \alpha_2 & \cdots & \alpha_0
\end{pmatrix}
=\begin{pmatrix} A & B \\ B & A \end{pmatrix},
\]
where
\[
A=\Circ\big(\alpha_0,\alpha_1,\dots,\alpha_{n-1}\big),\qquad
B=\RevCirc\big(\beta_0,\beta_1,\dots,\beta_{n-1}\big).
\]

\subsubsection{Form 2.}
With the listing $\{e,a,\dots,a^{n-1},b,ab,\dots,a^{n-1}b\}$, every element has the form
\[
c=\alpha_0+\alpha_1a+\cdots+\alpha_{n-1}a^{n-1}
+\beta_0b+\beta_1ab+\cdots+\beta_{n-1}a^{n-1}b.
\]
Proceeding as above (now $b\cdot a^i=a^{-i}b$), the generating matrix is
\[
\mathcal{G}_2=\begin{pmatrix} A & D \\[2pt] D^{\mathsf T} & A^{\mathsf T}\end{pmatrix},
\qquad
A=\Circ\big(\alpha_0,\dots,\alpha_{n-1}\big),\quad
D=\Circ\big(\beta_0,\dots,\beta_{n-1}\big).
\]

\begin{remark}{\em
In Section 4, we establish the relationship between these two matrices and the corresponding relationship between the group codes they generate. Further listings are possible; for instance $\{e,b,a,ba,\dots,a^{n-1},ba^{n-1}\}$ leads to a block-circulant matrix built from $n$ blocks of order $2\times 2$, which is less convenient to display. We therefore restrict attention to Forms~1 and~2, whose blocks are the two $n\times n$ matrices $A,B$ (Form~1) and $A,D$ (Form~2).
}\end{remark}

We now specialize the duality criteria of Subsection~\ref{subsec:duality} to the dihedral case. Recall from Lemma~\ref{lem:dichotomy} that every reflection $ba^i$ is an involution, so for $c=\sum_{i}\alpha_i a^i+\sum_i\beta_i ba^i$ one has $c^*=\sum_i\alpha_i a^{-i}+\sum_i\beta_i ba^i$.

\begin{corollary}\label{cor:dihedral-sd}
Let $\mathcal{G}_1=\sigma(c)=\big(\begin{smallmatrix}A&B\\B&A\end{smallmatrix}\big)$,
where \[A=\Circ(\alpha_0,\dots,\alpha_{n-1}),~~~~B=\RevCirc(\beta_0,\dots,\beta_{n-1}).\]
The code $\mathcal{C}(c)$ is self-orthogonal if and only if
\[
AA^{\mathsf T}+B^2=O \qquad\text{and}\qquad 2AB=O.
\]
In particular, if the characteristic of $\mathbb{F}_q$ is $2$, the second
condition is trivially satisfied, and $\mathcal{C}(c)$ is self-orthogonal
if and only if $AA^{\mathsf T}+B^2=O$.
\end{corollary}

\begin{proof}
Transposing the block matrix $\sigma(c)=\big(\begin{smallmatrix}A&B\\B&A\end{smallmatrix}\big)$
and using $B^{\mathsf T}=B$ (Lemma~\ref{lem:revcirc}) gives
\[
\sigma(c)^{\mathsf T}=\begin{pmatrix}A^{\mathsf T}&B^{\mathsf T}\\B^{\mathsf T}&A^{\mathsf T}\end{pmatrix}
=\begin{pmatrix} A^{\mathsf T}&B\\ B&A^{\mathsf T}\end{pmatrix}.
\]
Multiplying out gives
\[
\sigma(c)\sigma(c)^{\mathsf T}=
\begin{pmatrix} AA^{\mathsf T}+B^2 & AB+BA^{\mathsf T}\\[2pt]
AB+BA^{\mathsf T} & AA^{\mathsf T}+B^2\end{pmatrix}.
\]
By Lemma~\ref{lem:revcirc}, $B = CJ$ where $C$ is a circulant matrix and $J$
is the exchange matrix. Since $A$ is circulant, $JAJ=A^{\mathsf T}$, and as
$J^2=I$ this gives the identity $JA^{\mathsf T}=AJ$. Since $A$ and $C$ are
circulant they commute, so
\[
BA^{\mathsf T} = (CJ)A^{\mathsf T} = C(JA^{\mathsf T}) = C(AJ) = (CA)J = (AC)J
= A(CJ) = AB.
\]
Consequently, $AB+BA^{\mathsf T} = 2AB$. The claim then follows from
Theorem~\ref{thm:selfdual}(i).
\end{proof}
 
\begin{example}\label{ex:d4}{\em
Take $q=2$, $n=4$, and
\[
c=a^3+ba+ba^2+ba^3\in\mathbb{F}_2[D_4],
\]
i.e.\ $\alpha=(0,0,0,1)$ and $\beta=(0,1,1,1)$, so $A=\Circ(0,0,0,1)$ and $B=\RevCirc(0,1,1,1)$. A direct check gives $AA^{\mathsf T}=B^2=I$, hence $AA^{\mathsf T}+B^2=O$, and $AB=BA^{\mathsf T}$, hence $AB+BA^{\mathsf T}=O$ over $\mathbb{F}_2$, so by Corollary~\ref{cor:dihedral-sd} the code is self-orthogonal; here
\[
\mathcal{G}_1=\sigma(c)=
\begin{pmatrix}
0&0&0&1&0&1&1&1\\
1&0&0&0&1&1&1&0\\
0&1&0&0&1&1&0&1\\
0&0&1&0&1&0&1&1\\
0&1&1&1&0&0&0&1\\
1&1&1&0&1&0&0&0\\
1&1&0&1&0&1&0&0\\
1&0&1&1&0&0&1&0
\end{pmatrix},\qquad \operatorname{rank}\sigma(c)=4=\tfrac{1}{2}\,|D_4|.
\]
By Theorem~\ref{thm:selfdual}(ii) the row space $\mathcal{C}(c)$ is a self-dual $[8,4,4]$ binary code, with weight enumerator $1+14z^4+z^8$; this is the (unique) doubly even self-dual code of length $8$, that is, the extended Hamming code. Thus the dihedral group ring $\mathbb{F}_2[D_4]$ already realizes an optimal self-dual code through a single element satisfying $cc^*=0$.
}\end{example}

\begin{remark}{\em
The same machinery produces LCD codes via Theorem~\ref{thm:lcd}. For instance, taking $q=3$, $n=3$, and $c=e+a+2a^2+b+2ba+ba^2\in\mathbb{F}_3[D_3]$, i.e.\ coefficient vector $(1,1,2,1,2,1)$ in the Form~1 listing, gives $\operatorname{rank}\sigma(cc^*)=\operatorname{rank}\sigma(c)=3$, so the resulting ternary $[6,3,2]$ code is LCD.
}\end{remark}

\begin{remark}{\em
The generalized quaternion group of order $4n$,
\[
Q_{4n}=\langle a,b\mid a^{2n}=1,\ a^n=b^2,\ b^{-1}ab=a^{-1}\rangle,\qquad n\ge 2,
\]
(with $Q_8$ recovered for $n=2$) shares the relation $b^{-1}ab=a^{-1}$ with $D_n$. Its elements may be listed as $\{1,a,\dots,a^{2n-1},b,ab,\dots,a^{2n-1}b\}$ or as $\{1,a,\dots,a^{2n-1},b,ba,\dots,ba^{2n-1}\}$, and the generating matrices for $\mathbb{F}_q[Q_{4n}]$ are obtained by the same reasoning as for $D_n$.
}\end{remark}

\subsection{The Group Ring $\mathbb{F}_q[C_l\times C_m]$}

We start with the general case. Let $C_l\times C_m=\langle x,y\mid x^l=y^m=1,\ xy=yx\rangle$, of order $n=lm$, with elements $\{x^iy^j:0\le i\le l-1,\ 0\le j\le m-1\}$.

\subsubsection{Form 1.}
List the group as $C_l\times C_m=\{g_1,\dots,g_n\}$ with
\[
g_{1+i+lj}=x^iy^j,\qquad 0\le i\le l-1;\ 0\le j\le m-1,
\]
and let $f_1=\sum_{r=1}^n\alpha_r g_r$. Computing the coefficients of $g_rf_1$ row by row gives
\[
\mathcal{G}_1=\Circ(A_1,A_2,\dots,A_m),
\]
where the $A_s$ are the $l\times l$ circulants
\[
A_s=\Circ\big(\alpha_{1+l(s-1)},\alpha_{2+l(s-1)},\dots,\alpha_{l+l(s-1)}\big),\qquad
s=1,2,\dots,m.
\]

\subsubsection{Form 2.}
List the group with
\[
g_{1+j+mi}=x^iy^j,\qquad 0\le j\le m-1;\ 0\le i\le l-1,
\]
and let $f_2=\sum_{r=1}^n\alpha_r g_r$. Then
\[
\mathcal{G}_2=\Circ(A_1,A_2,\dots,A_l),
\]
where the $A_t$ are the $m\times m$ circulants
\[
A_t=\Circ\big(\alpha_{1+m(t-1)},\alpha_{2+m(t-1)},\dots,\alpha_{m+m(t-1)}\big),\qquad
t=1,2,\dots,l.
\]

\begin{remark}{\em
Since $C_l\times C_m$ is abelian, $\mathcal{C}(f)$ is an abelian code, i.e.\ an ideal of $\mathbb{F}_q[s,t]/(s^l-1,t^m-1)$. When $\gcd(lm,q)=1$ this ring is semisimple, and Lemma~\ref{lem:circ-algebra} together with the block-circulant structure reduces the computation of $\dim\mathcal{C}(f)=\operatorname{rank}\mathcal{G}$ to a product of one-variable gcd computations through the Fourier (CRT) decomposition.
}\end{remark}

\begin{remark}\label{rem:multi}{\em
The construction extends to any finite direct product of cyclic groups. For $C_k\times C_l\times C_m=\{g_1,\dots,g_n\}$ with $n=klm$, fixing the word $x^ay^bz^c$ and varying the order of the index triple $(a,b,c)$ yields six listings, e.g.
\[
g_{1+a+kb+klc}=x^ay^bz^c \quad\text{and}\quad g_{1+a+kc+klb}=x^ay^bz^c,
\]
and so on for the remaining four permutations. 
}\end{remark}

\begin{remark}{\em
For the semidirect product $C_l\rtimes C_m=\langle x,y\mid x^l=1,\ y^m=1,\ yxy^{-1}=x^k\rangle$ with $\gcd(k,l)=1$, the same procedure (now with the twist $yx=x^ky$ replacing commutativity) produces the corresponding generating matrices for fixed $l,m$.

Let
$$ C_l\rtimes C_m=\langle x,y\mid x^l=y^m=e,\ yxy^{-1}=x^k\rangle, \qquad \gcd(k,l)=1,\ \ k^m\equiv 1\pmod l, $$
a group of order $n=lm$. Taking $k\equiv1\pmod l$ recovers $C_l\times C_m$, and taking $m=2,\ k\equiv-1\pmod l$ recovers $D_l$.

Repeated use of $yx=x^ky$ gives, for every $v\ge0$,
$$ y^vx^i=x^{ik^v\bmod l}\,y^v,\qquad\text{equivalently}\qquad x^iy^v=y^v x^{ik^{-v}\bmod l}. $$

For $t\in\mathbb{Z}_l^\times$, the $t$-circulant $\Circ_t(\gamma_0,\dots,\gamma_{l-1})$ is the $l\times l$ matrix with $(r,c)$-entry $\gamma_{(c-tr)\bmod l}$. Thus $\Circ_1=\Circ$ and $\Circ_{-1}=\RevCirc$.

List $C_l\rtimes C_m=\{g_1,\dots,g_n\}$ with
$$ g_{1+i+lj}=y^jx^i,\qquad 0\le i\le l-1;\ 0\le j\le m-1, $$
and let $f=\sum_{j=0}^{m-1}\sum_{i=0}^{l-1}\alpha_i^{(j)}\,y^jx^i$. Using the twisting relation above,
$$ (y^vx^u)f=\sum_{j,i}\alpha_i^{(j)}\,y^{v+j}x^{\,(uk^{-j}+i)\bmod l}, $$
so re-indexing by $j'=(v+j)\bmod m$, $i'=(uk^{-j}+i)\bmod l$ shows that the coefficient of $y^{j'}x^{i'}$ in $(y^vx^u)f$ equals $\alpha^{(s)}_{(i'-uk^{-s})\bmod l}$, where $s=(j'-v)\bmod m$. Hence, writing
$$ B_s:=\Circ_{\,k^{-s}}\!\big(\alpha_0^{(s)},\alpha_1^{(s)},\dots,\alpha_{l-1}^{(s)}\big),\qquad s=0,1,\dots,m-1, $$
the generating matrix is the block circulant of these twisted blocks:
$$ \mathcal{G}_1=\Circ\big(B_0,B_1,\dots,B_{m-1}\big).$$

If $k\equiv1\pmod l$, every $k^{-s}=1$ and each $B_s$ is a circulant, giving exactly the Form~1 matrix of $C_l\times C_m$. If $m=2$ and $k\equiv-1\pmod l$, then $B_0=A:=\Circ(\alpha^{(0)})$ and $B_1=\RevCirc(\alpha^{(1)})=:B$, so
$$ \mathcal{G}_1=\begin{pmatrix}A&B\\B&A\end{pmatrix}, $$
recovering the Form~1 matrix of $D_l$.
}\end{remark}
\subsection{The Group Ring $\mathbb{F}_q[C_l\times D_m]$}
The number of generators and relations governs the number of natural listings. For $\mathbb{F}_q[C_l\times D_m]$ there are three generators and three relations, giving several listings. We first illustrate with $l=m=3$. Here $C_3=\{1,x,x^2\}$ and $D_3=\langle y,z\mid y^3=z^2=1,\ zyz^{-1}=y^{-1}\rangle$, so $D_3$ may be listed as $\{1,y,y^2,z,zy,zy^2\}$ or as $\{1,y,y^2,z,yz,y^2z\}$. Combining each with the listing of $C_3$ yields, among others, the four forms
\begin{align*}
&\text{Form 1:}\quad \{x^iz^jy^k:\ i=0,1,2;\ k=0,1,2;\ j=0,1\},\\
&\text{Form 2:}\quad \{z^jy^kx^i:\ k=0,1,2;\ j=0,1;\ i=0,1,2\},\\
&\text{Form 3:}\quad \{x^iy^kz^j:\ i=0,1,2;\ k=0,1,2;\ j=0,1\},\\
&\text{Form 4:}\quad \{y^kz^jx^i:\ k=0,1,2;\ j=0,1;\ i=0,1,2\}.
\end{align*}

In general we set
\[
C_l\times D_m=\langle x,y,z\mid x^l=y^m=z^2=1,\ xy=yx,\ xz=zx,\ zyz^{-1}=y^{-1}\rangle,
\]
a group of order $n=2ml$.

\subsubsection{Form 1.}
List $C_l\times D_m=\{g_1,\dots,g_n\}$ with
\[
g_{1+i+kl+jml}=x^iz^jy^k,\qquad 0\le i\le l-1;\ 0\le k\le m-1;\ 0\le j\le 1,
\]
and let $f_1=\sum_{r=1}^n\alpha_r g_r$. Since $x$ is central while $z$ inverts $y$, the $x$-coordinate produces $l\times l$ circulant blocks, the $y$-coordinate arranges them in a (reverse) circulant pattern, and the $z$-coordinate gives the outer order-two block structure. The result is
\[
\mathcal{G}_1=\begin{pmatrix} M_1 & M_2 \\[2pt] M_2 & M_1\end{pmatrix}=\Circ(M_1,M_2),
\]
where $M_1=\Circ(A_1,\dots,A_m)$ and $M_2=\RevCirc(A_{m+1},\dots,A_{2m})$, with the $A_s$ the $l\times l$ circulants
\[
A_s=\Circ\big(\alpha_{(s-1)l+1},\alpha_{(s-1)l+2},\dots,\alpha_{(s-1)l+l}\big),\qquad
s=1,2,\dots,2m.
\]

\begin{example}{\em
Take $l=m=3$ and the listing $g_{1+i+3k+9j}=x^i z^j y^k$, writing $r$ for $\alpha_r$. We obtain $\sigma(f_1)=\Circ(M_1,M_2)$ with $M_1=\Circ(A_1,A_2,A_3)$, $M_2=\RevCirc(A_4,A_5,A_6)$, and
\[
A_1=\Circ(1,2,3),\ A_2=\Circ(4,5,6),\ A_3=\Circ(7,8,9),
\]
\[
A_4=\Circ(10,11,12),\ A_5=\Circ(13,14,15),\ A_6=\Circ(16,17,18),
\]
i.e.\ the $18\times18$ matrix
\[
\sigma(f_1)=
\setlength{\arraycolsep}{2.6pt}\renewcommand{\arraystretch}{1.5}
\left(\begin{smallmatrix}
1&2&3&4&5&6&7&8&9&10&11&12&13&14&15&16&17&18\\
3&1&2&6&4&5&9&7&8&12&10&11&15&13&14&18&16&17\\
2&3&1&5&6&4&8&9&7&11&12&10&14&15&13&17&18&16\\
7&8&9&1&2&3&4&5&6&13&14&15&16&17&18&10&11&12\\
9&7&8&3&1&2&6&4&5&15&13&14&18&16&17&12&10&11\\
8&9&7&2&3&1&5&6&4&14&15&13&17&18&16&11&12&10\\
4&5&6&7&8&9&1&2&3&16&17&18&10&11&12&13&14&15\\
6&4&5&9&7&8&3&1&2&18&16&17&12&10&11&15&13&14\\
5&6&4&8&9&7&2&3&1&17&18&16&11&12&10&14&15&13\\
10&11&12&13&14&15&16&17&18&1&2&3&4&5&6&7&8&9\\
12&10&11&15&13&14&18&16&17&3&1&2&6&4&5&9&7&8\\
11&12&10&14&15&13&17&18&16&2&3&1&5&6&4&8&9&7\\
13&14&15&16&17&18&10&11&12&7&8&9&1&2&3&4&5&6\\
15&13&14&18&16&17&12&10&11&9&7&8&3&1&2&6&4&5\\
14&15&13&17&18&16&11&12&10&8&9&7&2&3&1&5&6&4\\
16&17&18&10&11&12&13&14&15&4&5&6&7&8&9&1&2&3\\
18&16&17&12&10&11&15&13&14&6&4&5&9&7&8&3&1&2\\
17&18&16&11&12&10&14&15&13&5&6&4&8&9&7&2&3&1
\end{smallmatrix}\right).
\]
}\end{example}

\subsubsection{Form 2.}
List $C_l\times D_m=\{g_1,\dots,g_n\}$ with
\[
g_{1+k+jm+2mi}=z^jy^kx^i,\qquad 0\le k\le m-1;\ 0\le j\le 1;\ 0\le i\le l-1,
\]
and let $f_2=\sum_{r=1}^n\alpha_r g_r$. Here the central factor $x$ is outermost, so the generating matrix is block circulant over $C_l$ with $D_m$-type blocks:
\[
\mathcal{G}_2=\Circ\big(M_1,M_2,\dots,M_l\big),\qquad M_s=\Circ\big(A_{2s-1},A_{2s}\big),
\quad s=1,2,\dots,l,
\]
where each $A_t, t=1,2,\dots,2l$, is an $m\times m$ matrix given by
\[
A_t=
\begin{cases}
\Circ\big(\alpha_{(t-1)m+1},\dots,\alpha_{(t-1)m+m}\big), & t \text{ odd},\\[4pt]
\RevCirc\big(\alpha_{(t-1)m+1},\dots,\alpha_{(t-1)m+m}\big), & t \text{ even}.
\end{cases}
\]
Each block $M_s$ is precisely the Form~1 dihedral matrix $\big(\begin{smallmatrix}A_{2s-1}&A_{2s}\\A_{2s}&A_{2s-1}\end{smallmatrix}\big)$ of $\mathbb{F}_q[D_m]$.

\subsubsection{Form 3.}
List $C_l\times D_m=\{g_1,\dots,g_n\}$ with
\[
g_{1+i+kl+jml}=x^iy^kz^j,\qquad 0\le i\le l-1;\ 0\le k\le m-1;\ 0\le j\le 1,
\]
and let $f_3=\sum_{r=1}^n\alpha_r g_r$. Then
\[
\mathcal{G}_3=\begin{pmatrix} M_1 & M_2 \\[3pt] M_2^{\mathsf T} & M_1^{\mathsf T}\end{pmatrix},
\]
where $M_1=\Circ(A_1,\dots,A_m)$, $M_2=\Circ(A_{m+1},\dots,A_{2m})$, and the $A_s$ are the $l\times l$ circulants
\[
A_s=\Circ\big(\alpha_{(s-1)l+1},\dots,\alpha_{(s-1)l+l}\big),\qquad s=1,2,\dots,2m.
\]
This is the $C_l$-blocked analogue of the Form~2 dihedral matrix $\big(\begin{smallmatrix}A&D\\D^{\mathsf T}&A^{\mathsf T}\end{smallmatrix}\big)$.

\begin{remark}{\em
For a block-circulant matrix $M=\Circ(B_1,B_2,B_3)$, the transpose is obtained by transposing the block arrangement \emph{and} each block:
\[
M^{\mathsf T}
=\begin{pmatrix} B_1 & B_2 & B_3 \\ B_3 & B_1 & B_2 \\ B_2 & B_3 & B_1\end{pmatrix}^{\!\mathsf T}
=\begin{pmatrix} B_1^{\mathsf T} & B_3^{\mathsf T} & B_2^{\mathsf T} \\
B_2^{\mathsf T} & B_1^{\mathsf T} & B_3^{\mathsf T} \\
B_3^{\mathsf T} & B_2^{\mathsf T} & B_1^{\mathsf T}\end{pmatrix}.
\]
When the $B_i$ are circulant, each $B_i^{\mathsf T}$ is again circulant as  \[\Circ(b_1,b_2,\dots,b_m)^{\mathsf T}=\Circ(b_1,b_m,\dots,b_2),\] so $M^{\mathsf T}$ is again block circulant with circulant blocks.
}\end{remark}

\subsubsection{Form 4.}
List $C_l\times D_m=\{g_1,\dots,g_n\}$ with
\[
g_{1+k+jm+2mi}=y^kz^jx^i,\qquad 0\le k\le m-1;\ 0\le j\le 1;\ 0\le i\le l-1,
\]
and let $f_4=\sum_{r=1}^n\alpha_r g_r$. Then
\[
\mathcal{G}_4=\Circ\big(M_1,M_2,\dots,M_l\big),\qquad
M_s=\begin{pmatrix} A_{2s-1} & A_{2s} \\[4pt] A_{2s}^{\mathsf T} & A_{2s-1}^{\mathsf T}\end{pmatrix},
\quad s=1,2,\dots,l,
\]
where the $A_t$ are the $m\times m$ circulants
\[
A_t=\Circ\big(\alpha_{(t-1)m+1},\dots,\alpha_{(t-1)m+m}\big),\qquad t=1,2,\dots,2l.
\]
Each block $M_s$ is the Form~2 dihedral matrix of $\mathbb{F}_q[D_m]$.

\begin{remark}{\em
We note that one form of the generating matrices for these block-circulant matrices arising from $D_n$, $C_l \times C_m$, and $C_l \times D_m$ has already appeared in \cite{YZ23,YZ24}. 
 }
\end{remark}

\begin{remark}{\em
The same argument yields generating matrices for $\mathbb{F}_q[C_l\times Q_{4m}]$, where $Q_{4m}$ is the generalized quaternion group of order $4m$, the only change being that the dihedral relation $zyz^{-1}=y^{-1}$ is replaced by the quaternion relations.
}\end{remark}

\section{Effect of the Ordering on the Generating Matrix}\label{sec:ordering}

Fix $a \in \mathbb{F}_q[G]$. Let $\sigma = (g_1,\ldots,g_n)$ and $\tau = (h_1,\ldots,h_n)$ be two orderings of $G$, giving generating matrices $M_\sigma$ and $M_\tau$ respectively, each in $M_n(\mathbb{F}_q)$.

\begin{lemma}[Permutation conjugation]\label{lem:perm}
Let $\pi:=\sigma^{-1}\circ\tau$, so that $h_i=g_{\pi(i)}$ for $1\le i\le n$, and let $P=(P_{ij})$ be the permutation matrix of $\pi$, i.e.\ $P_{ij}=1$ if $j=\pi(i)$ and $0$ otherwise. Then $M_\tau=PM_\sigma P^{-1}$.
\end{lemma}

\begin{proof}
The $(i,j)$-entry of $M_\sigma$ is $\alpha_{g_i^{-1}g_j}$ and that of $M_\tau$ is $\alpha_{h_i^{-1}h_j}$. Since $h_i=g_{\pi(i)}$,
\[
(M_\tau)_{ij}=\alpha_{g_{\pi(i)}^{-1}g_{\pi(j)}}.
\]
On the other hand, using $P_{ik}=1~\text{if and only if} ~k=\pi(i)$, and $(P^{-1})_{\ell j}=(P^{\mathsf T})_{\ell j}=P_{j\ell}=1~\text{if and only if}~ \ell=\pi(j)$,
\[
(PM_\sigma P^{-1})_{ij}=\sum_{k,\ell}P_{ik}(M_\sigma)_{k\ell}(P^{-1})_{\ell j}=(M_\sigma)_{\pi(i),\pi(j)}=\alpha_{g_{\pi(i)}^{-1}g_{\pi(j)}}=(M_\tau)_{ij}.
\]
Since the two matrices agree entrywise, $M_\tau=PM_\sigma P^{-1}$.
\end{proof}

\begin{theorem}[Effect of ordering on codes]\label{thm:ordering}
Let $\sigma$ and $\tau$ be two orderings of the same group $G$, and let $C_\sigma$ and $C_\tau$ be the linear codes generated by the Hurley matrices $M_\sigma$ and $M_\tau$, respectively, associated with a fixed element $a\in\mathbb{F}_q[G]$. 
\begin{enumerate}
    \item[\rm (1)] The codes $C_\sigma$ and $C_\tau$ are permutation equivalent. In particular, they have the same length, dimension, complete weight enumerator, and hence the same minimum distance.
    \item[\rm (2)] If $c\in C_\sigma$ is a minimum-weight codeword with support $S\subseteq\{1,\ldots,n\}$, then the corresponding minimum-weight codeword $cP^{-1}\in C_\tau$ has support $\pi^{-1}(S)$, where $\pi$ is the permutation associated with the change of ordering. Consequently, although the minimum distance is preserved, the coordinate positions of minimum-weight codewords may differ, and properties depending on those positions, such as symplectic self-orthogonality conditions, may also depend on the chosen ordering of $G$.
\end{enumerate}
\end{theorem}

\begin{proof}
(1) By Lemma~\ref{lem:perm}, $M_\tau=PM_\sigma P^{-1}$, so $C_\tau=\operatorname{RowSpace}(PM_\sigma P^{-1})$. Left multiplication by $P$ permutes the rows of $M_\sigma$ and does not change its row space; right multiplication by $P^{-1}$ permutes the columns, hence the coordinates of every codeword. Therefore, we have 
\[
C_\tau=\{\,cP^{-1}:c\in C_\sigma\,\}.
\]
We obtain that $C_\sigma$ and $C_\tau$ are permutation equivalent, with identical weight enumerators; in particular, $\dim C_\sigma=\dim C_\tau$ and $d(C_\sigma)=d(C_\tau)$. (The dimensions also agree directly, as $\rank M_\sigma=\rank(PM_\sigma P^{-1})=\rank M_\tau$.)

\medskip
(2) By Part~(1), $C_\tau=\{cP^{-1}:c\in C_\sigma\}$ with $P$ the permutation matrix of $\pi$. Let $c\in C_\sigma$ be a minimum-weight codeword with $S=\operatorname{supp}(c)$. Right multiplication by $P^{-1}$ permutes the coordinates, so
\[
\operatorname{supp}(cP^{-1})=\pi^{-1}(S),\qquad \operatorname{wt}(cP^{-1})=\operatorname{wt}(c),
\]
and $cP^{-1}$ is a minimum-weight codeword of $C_\tau$. The minimum distance is unchanged, but the location of the nonzero coordinates generally differ, so properties that depend on coordinate positions, such as certain self-orthogonality conditions, may vary with the ordering of $G$.
\end{proof}

\begin{corollary} \label{cor:listing}
For a fixed $a\in\mathbb{F}_q[G]$, any two orderings of $G$ produce permutation-equivalent group codes. In particular, all listings of $D_n$, $C_l\times C_m$, $C_k\times C_l\times C_m$, and $C_l\times D_m$ considered in Section~\ref{sec:gencodes} generate permutation-equivalent codes, with the same length, dimension, and complete weight enumerator.
\end{corollary}

\begin{proof}
Immediate from Theorem~\ref{thm:ordering}(1): any two orderings differ by a permutation $\pi$ of $G$, and the corresponding codes are related by the coordinate permutation $P^{-1}$.
\end{proof}

\begin{remark}\label{rem:ordering}{\em
The support positions of minimum-weight codewords depend on the ordering,
whereas the value of the minimum distance is an invariant of the abstract code.
Crucially, while Euclidean and Hermitian self-orthogonality
are invariant under coordinate permutations, symplectic
self-orthogonality is not, as the symplectic form $\Omega$ depends on the
exact bi-partition of the coordinates. This sensitivity to ordering is what
motivates the careful choice of listings in the symplectic quantum code
constructions of Section~\ref{sec:qecc}.
}\end{remark}

\section{Distinct Codes from Non-Isomorphic Groups}\label{sec:distinct}

We now show that non-isomorphic groups of the same order can yield genuinely inequivalent codes. We first record a structural lemma on group algebras.

\begin{lemma}[Wedderburn decomposition of group algebras]\label{lem:wedderburn_group_algebra}
Let $G$ be a finite group and $\mathbb{F}_q$ a finite field with $\gcd(|G|,q)=1$. By Maschke's theorem $\mathbb{F}_q[G]$ is semisimple, and by the Wedderburn--Artin theorem it decomposes uniquely (up to the order of the factors) as a direct product of matrix algebras over extension fields of $\mathbb{F}_q$:
\[
\mathbb{F}_q[G] \cong \prod_{i=1}^{r} M_{n_i}(\mathbb{F}_{q^{f_i}}).
\]
\end{lemma}

\begin{proof}
Since $\gcd(|G|,q)=1$, the characteristic of $\mathbb{F}_q$ does not divide $|G|$, so $\mathbb{F}_q[G]$ is semisimple by Maschke's theorem. The decomposition is the Wedderburn--Artin theorem for finite-dimensional semisimple algebras over a finite field, each simple component being a matrix algebra over a finite division ring, which is a field by Wedderburn's little theorem.
\end{proof}

The multiset of pairs $\mathcal{W}_G := \{(n_i,f_i)\}_{i=1}^r$ is a complete isomorphism invariant of $\mathbb{F}_q[G]$: two such algebras are isomorphic if and only if their multisets of pairs coincide.

\begin{theorem}[Non-isomorphic groups and code parameters]\label{thm:distinct}
Let $G_1$ and $G_2$ be groups of the same order $n$ with $\gcd(n,q)=1$, and let
\[
D_i=\bigl\{\dim_{\mathbb{F}_q}\mathcal{C}(a):a\in\mathbb{F}_q[G_i]\bigr\}
\]
be the set of dimensions realized by group codes over $\mathbb{F}_q[G_i]$. If $D_1\neq D_2$, then:
\begin{enumerate}
\item[\rm(1)] there exists $a_1\in\mathbb{F}_q[G_1]$ such that the parameters $[n,k,d]_q$ of $\mathcal{C}(a_1)$ cannot be realized by any code $\mathcal{C}(b)$ with $b\in\mathbb{F}_q[G_2]$ (or symmetrically for $G_2$ over $G_1$), and
\item[\rm(2)] the sets of achievable parameter pairs
\[
\bigl\{(k,d):\mathcal{C}(a)\text{ is an }[n,k,d]_q\text{ code},\ a\in\mathbb{F}_q[G_i]\bigr\}
\]
differ for $i=1$ and $i=2$.
\end{enumerate}
\end{theorem}

\begin{proof}
Since $\gcd(n,q)=1$, each $\mathbb{F}_q[G_i]$ is semisimple, so by Lemma~\ref{lem:wedderburn_group_algebra}
\[
\mathbb{F}_q[G_i]\cong\prod_j M_{n_{ij}}(\mathbb{F}_{q^{f_{ij}}}).
\]
For \(a\in\mathbb{F}_q[G_i]\) the group code \(\mathcal{C}(a)\) is the row space of \(\sigma(a)\), which is isomorphic to the principal left ideal \(\mathbb{F}_q[G_i]\,a\); hence \(\dim_{\mathbb{F}_q}\mathcal{C}(a)=\dim_{\mathbb{F}_q}\mathbb{F}_q[G_i]a\). In a single component \(M_n(\mathbb{F}_{q^f})\), the minimal left ideals are
the columns, each \(\cong \mathbb{F}_{q^f}^{\,n}\); a left ideal
isomorphic to \(c\) copies \((0\le c\le n)\) has
\(\mathbb{F}_{q^f}\)-dimension \(cn\), hence
\(\mathbb{F}_q\)-dimension \(cnf\). Since left ideals of a direct product are direct products of left ideals of the factors, the set of attainable dimensions over \(\mathbb{F}_q[G_i]\) is
\[
D_i=\Bigl\{\textstyle\sum_j c_j n_{ij} f_{ij} : 0\le c_j\le n_{ij}\Bigr\}.
\]
Now suppose $D_1\neq D_2$, and (without loss of generality) take $k\in D_1\setminus D_2$. Since $k\in D_1$, there is a left ideal $\mathbb{F}_q[G_1]a_1$ of dimension $k$, so $\mathcal{C}(a_1)$ has parameters $[n,k,d_1]_q$ for some $d_1$. If some $\mathcal{C}(b)$ with $b\in\mathbb{F}_q[G_2]$ had the same parameters, its dimension would be $k\in D_2$, a contradiction. Hence no group code over $\mathbb{F}_q[G_2]$ has parameters $[n,k,d_1]_q$, which proves~(1). Part (2) is then immediate: the pair $(k,d_1)$ is attainable over $\mathbb{F}_q[G_1]$ but not over $\mathbb{F}_q[G_2]$, so the attainable parameter sets differ.
\end{proof}

\begin{remark}\label{rem:dim_sets_insufficient}{\em
While Theorem~\ref{thm:distinct} establishes that distinct dimension sets
$D_i$ force distinct code parameters, the converse is false. It is entirely
possible for non-isomorphic groups to share identical dimension sets (i.e.,
their group algebras decompose into components of identical dimensions) while
still producing entirely distinct minimum distance profiles. This deeper
structural divergence cannot be detected by the Wedderburn decomposition alone
and serves as the primary motivation for the exhaustive computational comparisons
in Section~\ref{sec:comparison}.
}\end{remark}

\begin{example}[Groups of order $4$]\label{ex:z2z2}{\em
Let $n=4$ and $q=3$. The two groups of order $4$ are the cyclic group $G_1=C_4=\langle a\mid a^4=e\rangle$ and the Klein four-group $G_2=C_2\times C_2=\langle a,b\mid a^2=b^2=e,\ ab=ba\rangle$, and $\gcd(4,3)=1$.

Over $\mathbb{F}_3$ we have $x^4-1=(x-1)(x+1)(x^2+1)$ with $x^2+1$ irreducible, so
\[
\mathbb{F}_3[C_4]\cong\mathbb{F}_3\times\mathbb{F}_3\times\mathbb{F}_9,
\]
whereas $C_2\times C_2$ has exponent $2$ and four one-dimensional representations over $\mathbb{F}_3$, giving
\[
\mathbb{F}_3[C_2\times C_2]\cong\mathbb{F}_3\times\mathbb{F}_3\times\mathbb{F}_3\times\mathbb{F}_3.
\]
These algebras are non-isomorphic. Nevertheless, both realize the same set of code dimensions, $D_1=D_2=\{0,1,2,3,4\}$, and in fact the same parameter sets. For instance $1+a^2\in\mathbb{F}_3[C_4]$ has $\sigma(1+a^2)=\Circ(1,0,1,0)$ of rank $2$, giving a $[4,2,2]_3$ code, while $1+g$ (with $g\neq e$) over $C_2\times C_2$ also gives a $[4,2,2]_3$ code; both have the same weight enumerator $1+4z^2+4z^4$ and are permutation equivalent. A direct check at the remaining dimensions shows that the maximal distance attainable agrees for the two groups at every $k$.

This illustrates why the hypothesis $D_1\neq D_2$ of Theorem~\ref{thm:distinct} is needed: non-isomorphism of the group algebras alone does not force distinct parameters. Genuine differences appear once the distance profiles diverge; the order-$8$ comparison in Table~\ref{tab:gf3_n8} provides such an instance, where $C_8$ and $Q_8$ strictly outperform $C_2\times C_4$, $D_4$, and $C_2^3$ at dimensions $k=2,3,5$.
}\end{example}

\section{Comparison of Group Codes Across Group Structures}
\label{sec:comparison}
We now compare the group codes arising from every group of a fixed order. Let $G$ be a group of order $n$ and let $\mathbb{F}_q$ be a finite field. For $a = \sum_{g \in G} a_g\, g \in \mathbb{F}_q[G]$, the associated \emph{group code} $\mathcal{C}_a$ is the row space of the $n \times n$ matrix $M_a=\sigma(a)$, whose $i$-th row is the coefficient vector of $g_i a$ for a fixed ordering $g_1, \ldots, g_n$ of $G$. For each attainable dimension $k$, we record the best achievable minimum distance
\[
  d_{\max}(k, G, q)
  = \max\bigl\{\, d(\mathcal{C}_a) :\ a \in \mathbb{F}_q[G],\ \dim(\mathcal{C}_a) = k \,\bigr\}.
\]
All computations were performed in \textsc{Magma} \cite{Mag} by exhaustive search over all nonzero elements of $\mathbb{F}_q[G]$, taking one representative per scalar class.

\subsection{Ternary Codes of Length $8$: Groups of Order $8$}
\label{subsec:gf3_n8}

There are exactly five groups of order $8$ up to isomorphism: the cyclic group $C_8$, the direct product $C_2 \times C_4$, the dihedral group $D_4$, the quaternion group $Q_8$, and the elementary abelian group $C_2^3$. Table~\ref{tab:gf3_n8} reports $d_{\max}(k, G, 3)$ for each group; boldface indicates the largest minimum distance achieved for each dimension $k$.

\begin{table}[ht]
\centering
\caption{Best minimum distances for $[8,k,d]_3$ group codes over $\mathbb F_3$ for all five groups of order $8$. Boldface entries indicate the maximum distance achieved for that dimension across all five groups.}
\label{tab:gf3_n8}
\begin{tabular}{c|ccccc}
\toprule
$k$ & $C_8$ & $C_2 \times C_4$ & $D_4$ & $Q_8$ & $C_2^3$ \\
\midrule
1 & \textbf{8} & \textbf{8} & \textbf{8} & \textbf{8} & \textbf{8} \\
2 & \textbf{6} & 4           & 4           & \textbf{6} & 4          \\
3 & \textbf{5} & 4           & 4           & \textbf{5} & 4          \\
4 & \textbf{4} & \textbf{4}  & \textbf{4}  & \textbf{4} & \textbf{4} \\
5 & \textbf{3} & 2           & 2           & \textbf{3} & 2          \\
6 & \textbf{2} & \textbf{2}  & \textbf{2}  & \textbf{2} & \textbf{2} \\
7 & \textbf{2} & \textbf{2}  & \textbf{2}  & \textbf{2} & \textbf{2} \\
8 & \textbf{1} & \textbf{1}  & \textbf{1}  & \textbf{1} & \textbf{1} \\
\bottomrule
\end{tabular}
\end{table}

Two patterns emerge from Table~\ref{tab:gf3_n8}. First, $C_8$ and $Q_8$ produce \emph{identical} distance profiles and strictly outperform the other three groups at dimensions $k = 2, 3$, and $5$: they yield $[8,2,6]_3$ and $[8,3,5]_3$ codes and reach $d = 3$ at $k = 5$, whereas $C_2 \times C_4$, $D_4$, and $C_2^3$ attain only $d = 4$ at $k = 2, 3$ and $d = 2$ at $k = 5$. Second, the three groups $C_2 \times C_4$, $D_4$, and $C_2^3$ are mutually indistinguishable in their distance profiles. The coincidence of the cyclic group $C_8$ and the non-abelian group $Q_8$ shows that the  minimum-distance profile is \emph{not} determined solely by whether the group is abelian, but depends on finer properties of $\mathbb{F}_q[G]$. All five groups coincide at $k = 4$, each realizing the $[8,4,4]_3$ code.

\subsection{Binary Codes of Length $12$: Groups of Order $12$}
\label{subsec:gf2_n12}

The five groups of order $12$ up to isomorphism are the dicyclic group $C_3 \rtimes C_4$ (also written $\mathrm{Dic}_3$), the cyclic group $C_{12}$, the alternating group $A_4$, the dihedral group $D_6$, and the direct product $C_2 \times C_6$. Table~\ref{tab:gf2_n12} reports $d_{\max}(k, G, 2)$ for each group. A dash (---) indicates that no group ring element of that group realizes a code of the corresponding dimension; this is an algebraic constraint of $\mathbb{F}_2[G]$ and not a computational limitation.
\vskip 3pt
Here $\gcd(12,2)=2\neq1$, so $\mathbb{F}_2[G]$ is \emph{not} semisimple; Theorem~\ref{thm:distinct} (which assumes $\gcd(n,q)=1$) does not apply to this table. We therefore present it as a standalone computational comparison.

\begin{table}[ht]
\centering
\caption{Best minimum distances for $[12,k,d]_2$ group codes over $\mathbb F_2$ for all five groups of order $12$. Boldface entries indicate the maximum distance achieved for that dimension; ``---'' denotes dimensions not realizable by any group ring element of the respective group.}
\vskip 7pt
\label{tab:gf2_n12}
\begin{tabular}{c|ccccc}
\toprule
$k$ & $C_3 \rtimes C_4$ & $C_{12}$ & $A_4$ & $D_6$ & $C_2 \times C_6$ \\
\midrule
 1 & \textbf{12} & \textbf{12} & \textbf{12} & \textbf{12} & \textbf{12} \\
 2 &  \textbf{8} &  \textbf{8} &  \textbf{8} &  \textbf{8} &  \textbf{8} \\
 3 &  \textbf{6} &  \textbf{6} &  \textbf{6} &  4          &  4          \\
 4 &  \textbf{6} &  4          &  \textbf{6} &  4          &  \textbf{6} \\
 5 &  \textbf{4} &  \textbf{4} &  \textbf{4} &  \textbf{4} &  \textbf{4} \\
 6 &  \textbf{4} &  \textbf{4} &  \textbf{4} &  \textbf{4} &  \textbf{4} \\
 7 &  \textbf{4} &  \textbf{4} &  ---        &  \textbf{4} &  ---        \\
 8 &  \textbf{3} &  2          &  \textbf{3} &  2          &  \textbf{3} \\
 9 &  \textbf{2} &  \textbf{2} &  \textbf{2} &  \textbf{2} &  \textbf{2} \\
10 &  \textbf{2} &  \textbf{2} &  ---        &  \textbf{2} &  \textbf{2} \\
11 &  \textbf{2} &  \textbf{2} &  ---        &  ---        &  ---        \\
12 &  \textbf{1} &  \textbf{1} &  \textbf{1} &  \textbf{1} &  \textbf{1} \\
\bottomrule
\end{tabular}
\end{table}

\begin{remark} {\em Observations from the binary length-$12$ table
    \begin{enumerate}
        \item The dicyclic group $C_3\rtimes C_4$ and the alternating group $A_4$ each achieve $d = 6$ at \emph{both} $k=3$ and $k=4$, the highest attainable minimum distance at those dimensions. Among the remaining groups, the cyclic group $C_{12}$ reaches $d=6$ only at $k=3$ and $C_2\times C_6$ only at $k=4$, while the dihedral group $D_6$ does not exceed $d=4$ for any $k \geq 3$.
        \item  Only $C_3 \rtimes C_4$ and $C_{12}$ produce codes at \emph{every} dimension $k \in \{1,\ldots,12\}$. Missing dimensions reflect the idempotent structure of the respective group ring.
        \item The groups $C_3\rtimes C_4$, $A_4$, and $C_2\times C_6$ yield $[12,8,3]_2$ codes, whereas $C_{12}$ and $D_6$ achieve only $[12,8,2]_2$.
        \item The dicyclic group $C_3 \rtimes C_4$ gives the most complete and strongest performance: it covers all dimensions and matches or exceeds every other group in minimum distance for each $k$.
    \end{enumerate}
}\end{remark}

\begin{remark}{\em
Our computational search for $d_{\max}$
restricts attention to \emph{principal} left ideals $\mathcal{C}_a$ generated
by a single element $a \in \mathbb{F}_q[G]$. When $\gcd(|G|, q) = 1$, the
group algebra is semisimple and every left ideal is principal, making our
search theoretically exhaustive over all group codes. However, in
non-semisimple cases (such as Section~\ref{subsec:gf2_n12}), non-principal
left ideals may exist. Consequently, the distances reported in the non-semisimple
tables represent a lower bound on the absolute maximum distance achievable by
an arbitrary group code of that dimension.
}\end{remark}
\section{Quantum Codes from Self-Orthogonal Group Codes}\label{sec:qecc}
Recall that a group code over $\mathbb{F}_q$ is a left ideal of the group ring $\mathbb{F}_q[G]$. In this section we characterize the self-orthogonality of group codes under the Euclidean, Hermitian, and symplectic inner products, and construct quantum codes from each.

\subsection{QECCs from Euclidean Self-Orthogonal Group Codes}

\begin{proposition}\label{eu-self}
Let $M$ be a generating matrix for a linear code $C$. Then $C\subseteq C^{\perp_e}$ if and only if $MM^{\mathsf T} = O$, where $O$ is the zero matrix.
\end{proposition}

\begin{proof}
The rows of $M$ span $C$, so $C\subseteq C^{\perp_e}$ holds if and only if every row of $M$ is orthogonal to every codeword. Since the codewords are the $\mathbb{F}_q$-combinations of the rows, this is equivalent to $M r^{\mathsf T} = 0$ for each row $r$ of $M$; collecting these equations over all rows is exactly $M M^{\mathsf T} = O$.
\end{proof}

Using this, we obtain a necessary and sufficient condition for a group code to be Euclidean self-orthogonal.

\begin{corollary}\label{es}
Let $G$ be a finite group of order $n$ and $a\in\mathbb{F}_q[G]$. Suppose
$C=\mathcal{C}(a)$ is the group code with generating matrix $\sigma(a)$.
Then $C \subseteq C^{\perp_e}$ if and only
if $aa^* = 0$ in $\mathbb{F}_q[G]$.
\end{corollary}

\begin{proof}
This can also be proved by following the argument used in the proof of Theorem~\ref{thm:selfdual}(i). For completeness, we provide a direct proof.

Let $M=\sigma(a)$. By Lemma~\ref{lem:involution}, we have
\[
M^{\mathsf T}=\sigma(a)^{\mathsf T}=\sigma(a^*).
\]
Using Proposition~\ref{eu-self} together with the fact that $\sigma$ is a ring homomorphism, we obtain
\[
C\subseteq C^{\perp_e}
    \iff \sigma(a)\sigma(a^*)=O
    \iff \sigma(aa^*)=O
    \iff aa^*=0,
\]
where the final equivalence follows from the injectivity of $\sigma$.
\end{proof}

\begin{corollary}[\cite{DGTT18}]
If $a=a^*$ and $a^2=0$, then $\mathcal{C}(a)$ is Euclidean self-orthogonal.
\end{corollary}

A QECC with length $n$, dimension $k$, and minimum distance $d$ over $\mathbb{F}_q$ is denoted by $[[n,k,d]]_q$. We use the following result of \cite{AK01}, which constructs quantum codes from self-orthogonal linear codes.

\begin{theorem}[\cite{AK01}]\label{e1}
Let $C$ be an $[n, k, d]_q$ linear code over $\mathbb{F}_q$. If $C\subseteq C^{\perp_e}$, then there exists a QECC with parameters $[[n, n-2k, \ge d_H]]_q$, where $d_H$ is the minimum Hamming weight of $C^{\perp_e} \setminus C$.
\end{theorem}

Combining Theorem~\ref{e1} with Proposition~\ref{es} gives our main construction.

\begin{theorem}\label{T1}
Let $G$ be a finite group of order $n$ and $a\in\mathbb{F}_q[G]$. Suppose $C=\mathcal{C}(a)$ is the group code with parameters $[n,k,d_H]_q$. If $aa^*=0$, then there exists a QECC with parameters $[[n,n-2k,\ge d_H]]_q$.
\end{theorem}

\begin{proof}
If $aa^* = 0$, then $C$ is Euclidean self-orthogonal by Proposition~\ref{es}, and the result follows from Theorem~\ref{e1}.
\end{proof}

\begin{example}{\em
Consider the group ring $\mathbb{F}_q[C_l \times C_m]$ with $q = 2$, $l = 5$, and $m = 3$. Let $C_5 = \langle x\rangle$ and $C_3 = \langle y\rangle$, so that
\[
C_5 \times C_3 = \{e, y, y^2, x, xy, xy^2, x^2, x^2y, x^2y^2, x^3, x^3y, x^3y^2, x^4, x^4y, x^4y^2\}.
\]
Take $a \in \mathbb{F}_2[C_5 \times C_3]$ defined by
\[
a = e + y^2 + x^2 + x^2y^2 + x^3y + x^3y^2 + x^4y + x^4y^2,
\]
so that
\[
a^* = e + y + xy + xy^2 + x^2y + x^2y^2 + x^3 + x^3y.
\]
Let $C$ be the group code generated by $\sigma(a)$. Then $C$ is a $[15, 4, 8]_2$ linear code. A direct calculation gives $aa^* = 0$, so by Proposition~\ref{es} the code $C$ is Euclidean self-orthogonal. Theorem~\ref{T1} then yields a QECC with parameters $[[15, 7, 3]]_2$, which match the best known parameters of length $15$ recorded in \cite{Gra}. (Since $\gcd(5,3)=1$, this case is equivalent to $\mathbb{F}_2[C_{15}]$.)
}\end{example}

\subsection{An Infinite Family of Quantum Codes}\label{sec:family}

\begin{proposition}\label{prop:kron-iso}
Let $H$ and $K$ be finite groups. The $\mathbb{F}_q$-linear map
\[
\Theta:\mathbb{F}_q[H]\otimes_{\mathbb{F}_q}\mathbb{F}_q[K]\longrightarrow\mathbb{F}_q[H\times K],
\qquad h\otimes k\longmapsto (h,k),
\]
is an isomorphism of $\mathbb{F}_q$-algebras. Under this identification, for $a=\sum_{h}\alpha_h h\in\mathbb{F}_q[H]$ and $b=\sum_{k}\beta_k k\in\mathbb{F}_q[K]$ one has
\[
a\otimes b=\sum_{h\in H}\sum_{k\in K}\alpha_h\beta_k\,(h,k),\qquad
(a\otimes b)(a'\otimes b')=aa'\otimes bb',\qquad
(a\otimes b)^*=a^*\otimes b^*.
\]
\end{proposition}

\begin{proof}
The map $\Theta$ sends the basis $\{h\otimes k:h\in H,\ k\in K\}$ bijectively onto the basis $\{(h,k):h\in H,\ k\in K\}$ of $\mathbb{F}_q[H\times K]$, so it is an $\mathbb{F}_q$-vector space isomorphism. On basis elements,
\[
\Theta\bigl((h\otimes k)(h'\otimes k')\bigr)=\Theta\bigl(hh'\otimes kk'\bigr)=(hh',kk')=(h,k)(h',k')=\Theta(h\otimes k)\,\Theta(h'\otimes k'),
\]
and extending bilinearly shows $\Theta$ is multiplicative; hence $\Theta$ is an algebra isomorphism. The displayed formula for $a\otimes b$ is immediate, and
\[
(a\otimes b)(a'\otimes b')=\sum_{h,h',k,k'}\alpha_h\alpha'_{h'}\beta_k\beta'_{k'}(hh',kk')=aa'\otimes bb'.
\]
Finally, since $(h,k)^{-1}=(h^{-1},k^{-1})$,
\[
(a\otimes b)^*=\sum_{h,k}\alpha_h\beta_k\,(h^{-1},k^{-1})=\Bigl(\sum_h\alpha_h h^{-1}\Bigr)\otimes\Bigl(\sum_k\beta_k k^{-1}\Bigr)=a^*\otimes b^*. \qedhere
\]
\end{proof}

\begin{remark}{\em
In Subsection~\ref{subsec:duality} and the discussion of $\mathbb{F}_q[C_l\times C_m]$, we wrote the basis element $(x^i,y^j)$ simply as $x^iy^j$. This is justified by Proposition~\ref{prop:kron-iso}: under $\mathbb{F}_q[C_l]\otimes_{\mathbb{F}_q}\mathbb{F}_q[C_m]\cong\mathbb{F}_q[C_l\times C_m]$ we have $x^i\otimes y^j\mapsto (x^i,y^j)\mapsto x^iy^j$, and we freely identify these throughout.
}\end{remark}

\begin{lemma}\label{lem:kron}
Order the elements of $H\times K$ as $g_{(i,j)}=(h_i,k_j)$. Then $\sigma(a\otimes b)=\sigma(a)\otimes\sigma(b)$, where $\otimes$ on the right is the Kronecker product of matrices. Consequently, $\mathcal{C}(a\otimes b)=\mathcal{C}(a)\otimes\mathcal{C}(b)$ (the Kronecker product of codes) and
\[
\dim\mathcal{C}(a\otimes b)=\operatorname{rank}\sigma(a)\,\operatorname{rank}\sigma(b).
\]
\end{lemma}

\begin{proof}
The $((i,j),(i',j'))$-entry of $\sigma(a\otimes b)$ is the coefficient of $(h_i,k_j)^{-1}(h_{i'},k_{j'})=(h_i^{-1}h_{i'},\,k_j^{-1}k_{j'})$ in $a\otimes b$, namely
\[
\alpha_{h_i^{-1}h_{i'}}\beta_{k_j^{-1}k_{j'}}=\sigma(a)_{i,i'}\,\sigma(b)_{j,j'}=(\sigma(a)\otimes\sigma(b))_{(i,j),(i',j')}.
\]
Hence $\sigma(a\otimes b)=\sigma(a)\otimes\sigma(b)$. The row space of a Kronecker product $A\otimes B$ is the Kronecker product of the row spaces, so $\mathcal{C}(a\otimes b)=\mathcal{C}(a)\otimes\mathcal{C}(b)$, and $\operatorname{rank}(A\otimes B)=\operatorname{rank}(A)\operatorname{rank}(B)$ gives the stated dimension.
\end{proof}

\begin{theorem}\label{thm:product}
Let $\mathcal{C}(a)$ be an $[l,k_1,d_1]_q$ group code in $\mathbb{F}_q[H]$ and $\mathcal{C}(b)$ an $[m,k_2,d_2]_q$ group code in $\mathbb{F}_q[K]$, both  nonzero $($and, for parts $(iii)-(iv)$, proper$)$, and let $C=\mathcal{C}(a\otimes b)$ have length $lm$. Then:
\begin{enumerate}
\item[\rm(i)] $C$ has parameters $[\,lm,\ k_1k_2,\ d_1d_2\,]_q$;
\item[\rm(ii)] $(a\otimes b)(a\otimes b)^{*}=(aa^{*})\otimes(bb^{*})$, so $C\subseteq C^{\perp_e}$ if and only if  $aa^{*}=0$ or $bb^{*}=0$;
\item[\rm(iii)] $d(C^{\perp_e})=\min\bigl(d(\mathcal{C}(a)^{\perp_e}),\,d(\mathcal{C}(b)^{\perp_e})\bigr)$;
\item[\rm(iv)] if $aa^{*}=0$ $($or $bb^{*}=0)$ there is a QECC $\bigl[\!\bigl[\,lm,\ lm-2k_1k_2,\ \ge\min(d(\mathcal{C}(a)^{\perp_e}),d(\mathcal{C}(b)^{\perp_e}))\,\bigr]\!\bigr]_q$, with equality in the distance whenever $d_1,d_2\ge 2$.
\end{enumerate}
\end{theorem}

\begin{proof}
Write $U=\mathcal{C}(a)$ and $V=\mathcal{C}(b)$.

(i) $\dim C=k_1k_2$ by Lemma~\ref{lem:kron}. Writing codewords as $l\times m$ arrays, $U\otimes V=(U\otimes\mathbb{F}_q^m)\cap(\mathbb{F}_q^l\otimes V)=\{M:\text{columns}\in U,\ \text{rows}\in V\}$; a nonzero $M$ has a nonzero row (weight $\ge d_2$), hence at least $d_2$ nonzero columns, each of weight $\ge d_1$, so $\wt(M)\ge d_1d_2$, attained on an outer product of minimum-weight words.

(ii) By Proposition~\ref{prop:kron-iso}, $(a\otimes b)^*=a^*\otimes b^*$, and a Kronecker product of two algebra elements is zero if and only if  one factor is zero; combine with Proposition~\ref{es}.

(iii) We know that $C=U\otimes V=(U\otimes\mathbb{F}_q^m)\cap(\mathbb{F}_q^l\otimes V)$. Since $(A\cap B)^{\perp_e}=A^{\perp_e}+B^{\perp_e}$ for subspaces $A,B$ of a common ambient space, and $(U\otimes\mathbb{F}_q^m)^{\perp_e}=U^{\perp_e}\otimes\mathbb{F}_q^m$, $(\mathbb{F}_q^l\otimes V)^{\perp_e}=\mathbb{F}_q^l\otimes V^{\perp_e}$, we get
\[
C^{\perp_e}=(U^{\perp_e}\otimes\mathbb{F}_q^m)+(\mathbb{F}_q^l\otimes V^{\perp_e}).
\]
For a minimum-weight $u\in U^{\perp_e}$ and $e_1=(1,0,\ldots,0)\in\mathbb{F}_q^m$, the array $u\otimes e_1\in C^{\perp_e}$ has its first column equal to $u$ and all other columns zero, so $\wt(u\otimes e_1)=\wt(u)=d(U^{\perp_e})$; hence $d(C^{\perp_e})\le d(U^{\perp_e})$, and symmetrically $d(C^{\perp_e})\le d(V^{\perp_e})$. Conversely, let $0\neq M\in C^{\perp_e}$ with $w=\wt(M)$. Since $C=U\otimes V$, $\langle M,u\otimes v\rangle=u^{\mathsf T}Mv=0$ for all $u\in U$, $v\in V$, so $Mv\in U^{\perp_e}$ for every $v\in V$. As $Mv$ is supported on the nonzero rows of $M$, $\wt(Mv)\le w$. If $w<d(U^{\perp_e})$ then $Mv=0$ for all $v$, so every row of $M$ lies in $V^{\perp_e}$ and $w\ge d(V^{\perp_e})$. Interchanging $U,V$ gives the symmetric statement, so $w\ge\min(d(U^{\perp_e}),d(V^{\perp_e}))$. Combining the two inequalities yields $d(C^{\perp_e})=\min(d(U^{\perp_e}),d(V^{\perp_e}))$.

(iv) If $aa^*=0$ (or $bb^*=0$), then $C\subseteq C^{\perp_e}$ by (ii), and Theorem~\ref{T1} gives a QECC $[[lm,\,lm-2k_1k_2,\,d_Q]]_q$ with $d_Q=\wt(C^{\perp_e}\setminus C)\ge d(C^{\perp_e})=\min(d(U^{\perp_e}),d(V^{\perp_e}))$ by (iii). If $d_2\ge2$, a minimum-weight $u\in U^{\perp_e}$ gives $u\otimes e_1\in C^{\perp_e}$ with $e_1\notin V$ (every nonzero codeword of $V$ has weight at least 2), so $u\otimes e_1\notin U\otimes V=C$ and $d_Q\le\wt(u\otimes e_1)=d(U^{\perp_e})$. If $d_1\ge2$, symmetrically $d_Q\le d(V^{\perp_e})$. Hence $d_Q=\min(d(U^{\perp_e}),d(V^{\perp_e}))$ whenever $d_1,d_2\ge2$.
\end{proof}

\begin{corollary}\label{cor:famA}
For even $l,m$, the repetition codes $\mathcal{C}(a)=U$ and $\mathcal{C}(b)=V$ over $\mathbb{F}_2$ give a self-orthogonal $[lm,1,lm]_2$ code and a QECC $[[\,lm,\ lm-2,\ 2\,]]_2$.
\end{corollary}

\begin{proof}
For even $l$, the all-ones generator of $U$ satisfies $aa^*=0$ over $\mathbb{F}_2$, so $\mathcal{C}(a\otimes b)$ is self-orthogonal by Theorem~\ref{thm:product}(ii) and has parameters $[lm,1,lm]_2$ by~(i). The dual is the even-weight code with parameters $[lm,lm-1,2]$, so $d(U^{\perp_e})=d(V^{\perp_e})=2$; since $d(U)=l\ge2$ and $d(V)=m\ge2$, Theorem~\ref{thm:product}(iv) gives the QECC $[[lm,\,lm-2,\,2]]_2$.
\end{proof}

\begin{corollary}\label{cor:famB}
Let \(q\) be an odd prime power and take \(H=K=C_{q-1}\), so
\(G=C_{q-1}\times C_{q-1}\). Since \(q=p^s\) for a prime \(p\) with \(p\nmid q-1\),
\(\mathbb{F}_q[C_{q-1}]\) is semisimple (Maschke) and every cyclic code of length
\(q-1\) is a principal ideal, hence a group code \(C(a)\). Moreover, since
\(q\equiv 1\pmod{q-1}\), every \(q\)-cyclotomic coset mod \(q-1\) is a singleton,
so any subset \(Z\subseteq\mathbb{Z}_{q-1}\) is a valid defining set for a cyclic
code over \(\mathbb{F}_q\) (in particular the consecutive set below defines an
MDS Reed--Solomon code). Fix a primitive element \(\alpha\in\mathbb{F}_q^\times\) generating \(C_{q-1}\), and identify a cyclic code
with the polynomial \(c(x)=\sum c_i x^i\); its defining set is
\(Z=\{i\in\mathbb{Z}_{q-1}: c(\alpha^i)=0\}\). Let
\(C(a)=C(b)\) be the Reed--Solomon code of length \(q-1\) over
\(\mathbb{F}_q\) with defining set
\[
Z=\left\{0,1,\ldots,\frac{q-1}{2}\right\},
\] which is the MDS self-orthogonal code $[\,q-1,\tfrac{q-3}{2},\tfrac{q+3}{2}\,]_q$ with $d(\mathcal{C}(a)^{\perp_e})=\tfrac{q-1}{2}$. Then Theorem~\ref{thm:product} yields, for all $q\ge 5$,
\[
\Bigl[\!\Bigl[\,(q-1)^2,\ \tfrac{q^{2}+2q-7}{2},\ \tfrac{q-1}{2}\,\Bigr]\!\Bigr]_q ,
\]
an infinite family of QECCs of strictly increasing minimum distance.
\end{corollary}

\begin{proof}
We first show that $\mathcal{C}(a)$ is Euclidean self-orthogonal and then determine its parameters in order to apply Theorem~\ref{thm:product}. Let $n=q-1$. We first show that
\[
Z\cup(-Z)=\mathbb Z_n.
\]
For $j\in Z$ with $1\le j\le\tfrac{q-1}{2}$ we have $-j\equiv n-j=(q-1)-j\pmod n$, and as $j$ runs over $1,\dots,\tfrac{q-1}{2}$ the value $(q-1)-j$ runs over the integers in $[\tfrac{q-1}{2},\,q-2]$; together with $0$ (from $j=0$),
\[
-Z=\{0\}\cup\Bigl\{\tfrac{q-1}{2},\,\tfrac{q-1}{2}+1,\dots,q-2\Bigr\}.
\]
Since $Z=[0,\tfrac{q-1}{2}]\cap\mathbb{Z}$ and $-Z=\{0\}\cup\bigl([\tfrac{q-1}{2},q-2]\cap\mathbb{Z}\bigr)$ share the endpoint $\tfrac{q-1}{2}$ with no gap,
\[
Z\cup(-Z)=\{0,1,\dots,q-2\}=\mathbb{Z}_{q-1}.
\]
For a cyclic code with defining set $T$, the Euclidean dual has defining set $\mathbb{Z}_n\setminus(-T)$, and $C\subseteq C^{\perp_e}$ if and only if  $T\supseteq\mathbb{Z}_n\setminus(-T)$, i.e.\ if and only if  $\mathbb{Z}_n=T\cup(-T)$. Applying this with $T=Z$ gives $\mathcal{C}(a)\subseteq\mathcal{C}(a)^{\perp_e}$, so $aa^{*}=0$ by Proposition~\ref{es}.
\vskip 3pt

Since $|Z|=\tfrac{q+1}{2}$, the code $\mathcal{C}(a)=\mathcal{C}(b)$ is the MDS $[\,q-1,\tfrac{q-3}{2},\tfrac{q+3}{2}\,]_q$ code, so $k_1=k_2=\tfrac{q-3}{2}$ and $d_1=d_2=\tfrac{q+3}{2}$. Its dual is also MDS with $d(\mathcal{C}(a)^{\perp_e})=(q-1)-\tfrac{q+1}{2}+1=\tfrac{q-1}{2}$. By Theorem~\ref{thm:product}(iii), $d(C^{\perp_e})=\tfrac{q-1}{2}$, and
\[
lm-2k_1k_2=(q-1)^2-2\Bigl(\tfrac{q-3}{2}\Bigr)^2=\frac{q^2+2q-7}{2}.
\]
The constituent code is nonzero precisely when
\(k_1=\frac{q-3}{2}\ge 1\), i.e., \(q\ge 5\). For \(q\ge 5\) we further
have \(d_1=d_2=\frac{q+3}{2}\ge 2\), so
Theorem~\ref{thm:product}(iv) applies with equality and gives the
quantum code
\[
\Bigl[\!\Bigl[\,(q-1)^2,\ \tfrac{q^2+2q-7}{2},\ \tfrac{q-1}{2}\,\Bigr]\!\Bigr]_q.
\]
As $q\to\infty$, both the length $(q-1)^2$ and the distance $\tfrac{q-1}{2}$ grow without bound, giving an infinite family of QECCs of strictly increasing minimum distance.
\end{proof}

\begin{corollary}[Dihedral family]\label{cor:dihedral}
Let $C_l=\langle x\mid x^l=1\rangle$, $D_m=\langle r,s\mid r^m=s^2=1,\,srs=r^{-1}\rangle$,
and $G=C_l\times D_m$ of order $2lm$. Let $u\in\mathbb{F}_q[C_l]$ satisfy
$uu^*=0$ such that $\mathcal{C}(u)$ is a proper, nonzero $[l,k_1,d_1]_q$ code.
Let $v\in\mathbb{F}_q[D_m]$ be any non-zero, non-unit element with
$\mathcal{C}(v)=[2m,k_2,d_2]_q$. Then $\mathcal{C}(u\otimes v)$ is a
self-orthogonal $[2lm,k_1k_2,d_1d_2]_q$ code, and by Theorem~\ref{thm:product}
there is a QECC
\[
\bigl[\!\bigl[\,2lm,\ 2lm-2k_1k_2,\ \ge \min(d(\mathcal{C}(u)^{\perp_e}),
d(\mathcal{C}(v)^{\perp_e}))\,\bigr]\!\bigr]_q,
\]
with equality holding in the minimum distance whenever $d_1, d_2 \ge 2$.
\end{corollary}

\begin{example}{\em
Take $q=2$, $l=2$, $m=3$ (so $C_2\times D_3\cong D_6$). With $u=1+x$ we have $uu^*=0$, $\mathcal{C}(u)=[2,1,2]_2$, and $\sigma(u)=\begin{psmallmatrix}1&1\\1&1\end{psmallmatrix}$. With $v=1+r+r^2\in\mathbb{F}_2[D_3]$ (Form~1) we get $A=\Circ(1,1,1)=J_3$, $B=0$, so $\sigma(v)=\begin{psmallmatrix}J_3&0\\0&J_3\end{psmallmatrix}$ and $\mathcal{C}(v)=[6,2,3]_2$. Then $u\otimes v=\sum_{h\in\{1,x\}}\sum_{k\in\{1,r,r^2\}}(h,k)$, $\sigma(u\otimes v)=\sigma(u)\otimes\sigma(v)$, and $\mathcal{C}(u\otimes v)=[12,2,6]_2$ is self-orthogonal, giving the QECC $[[12,8,2]]_2$.
}\end{example}

\begin{lemma}[Distance ceiling]\label{lem:ceiling}
Let $C\subseteq\mathbb{F}_q^{N}$ be a Euclidean self-orthogonal linear code.
Then
\[
  d(C^{\perp_e})\le \dim C+1\le\Big\lfloor \tfrac{N}{2}\Big\rfloor+1 .
\]
If, in addition, $q$ is an odd prime power and $C$ is cyclic of length
$N=q-1$, then $\dim C\le\frac{q-3}{2}$, and hence
$d(C^{\perp_e})\le\frac{q-1}{2}$.
\end{lemma}

\begin{proof}
To establish the general bound, we note that for
 any linear code the Euclidean dual has complementary dimension,
$\dim C^{\perp_e}=N-\dim C$. Euclidean self-orthogonality
$C\subseteq C^{\perp_e}$ gives $\dim C\le\dim C^{\perp_e}$, whence
$2\dim C\le N$ and, as $\dim C\in\mathbb{Z}$, $\dim C\le\lfloor N/2\rfloor$.
Applying the Singleton bound to $C^{\perp_e}$,
\[
  d(C^{\perp_e})\le N-\dim C^{\perp_e}+1
              = N-(N-\dim C)+1
              = \dim C+1
              \le \Big\lfloor \tfrac{N}{2}\Big\rfloor+1 .
\]

In the cyclic case, 
let $q$ be odd and let $C$ be cyclic of length $N=q-1$. Since
$\gcd(N,q)=\gcd(q-1,q)=1$, the polynomial $x^{N}-1$ is separable over
$\mathbb{F}_q$, so $C$ is determined by its defining set
$T:=Z(C)\subseteq\mathbb{Z}_N$, and $\dim C=N-|T|$. The Euclidean dual is the
cyclic code with defining set $Z(C^{\perp_e})=\mathbb{Z}_N\setminus(-T)$; hence
\[
  C\subseteq C^{\perp_e}
  \iff Z(C^{\perp_e})\subseteq Z(C)
  \iff T\cup(-T)=\mathbb{Z}_N .
\]
Because $q$ is odd, $N=q-1$ is even, and the equation $x\equiv -x\pmod N$ has
exactly the two solutions $x=0$ and $x=N/2=(q-1)/2$. As $-0=0$ and
$-\tfrac{N}{2}=\tfrac{N}{2}$, covering these two points by $T\cup(-T)$ forces
$0,\tfrac{q-1}{2}\in T$. The remaining $N-2$ elements split into
$(N-2)/2$ disjoint pairs $\{x,-x\}$ with $x\ne -x$, and $T$ must contain at
least one element of each pair. Therefore
\[
  |T|\ \ge\ 2+\frac{N-2}{2}\ =\ \frac{N+2}{2}\ =\ \frac{q+1}{2},
\]
so that
\[
  \dim C = N-|T| \le (q-1)-\frac{q+1}{2}=\frac{q-3}{2}.
\]
Combining with the general bound $d(C^{\perp_e})\le\dim C+1$ gives
$d(C^{\perp_e})\le\frac{q-3}{2}+1=\frac{q-1}{2}$.
\end{proof}

\subsection{QECCs from Hermitian Self-Orthogonal Group Codes}

Throughout this subsection $q=p^2$ for a prime power $p$. For $\mathbf{a} = (a_1,\ldots,a_{n})$ and $\mathbf{b} = (b_1,\ldots,b_{n})$ in $\mathbb{F}_{q}^{n}$, the Hermitian inner product is
\[
\langle \mathbf{a}, \mathbf{b}\rangle_h=a_1b^p_{1}+ \cdots +a_{n}b^p_{n},
\]
and the Hermitian dual of $C$ is
\[
C^{\perp_h}= \{\mathbf{c} \in \mathbb{F}_{p^2}^{n} : \langle\mathbf{c},\mathbf{d}\rangle_h= 0 \text{ for all } \mathbf{d} \in C\}.
\]
We say $C$ is Hermitian self-orthogonal if $C \subseteq C^{\perp_h}$.

\vskip 5pt
For $a = \sum_{i=1}^n \alpha_i g_i \in \mathbb{F}_q[G]$ we set $a^p = \sum_{i=1}^n \alpha_i^p g_i \in \mathbb{F}_q[G]$, and for a matrix $M = (m_{ij})$ we set $\overline{M} = (m_{ij}^p)$.

\begin{proposition}\label{he-self}
Let $M$ be a generating matrix of a linear code $C$. Then $C \subseteq C^{\perp_h}$ if and only if $M \overline{M}^{\mathsf T} = O$, where $O$ is the zero matrix.
\end{proposition}

\begin{proof}
Write each codeword as $x = uM$ for $u \in \mathbb{F}_{q}^k$. For $x = uM$ and $y = vM$,
\[
\langle x, y \rangle_h = (uM)\,\overline{(vM)}^{\mathsf T} = u\,M \overline{M}^{\mathsf T}\,\overline{v}^{\mathsf T}.
\]
This vanishes for all $u, v$ if and only if $M \overline{M}^{\mathsf T} = O$.
\end{proof}

\begin{proposition} \label{H0}
Let $G$ be a finite group of order $n$ and $a\in\mathbb{F}_q[G]$. Suppose $C=\mathcal{C}(a)$ is the group code with generating matrix $\sigma(a)$. Then $C \subseteq C^{\perp_h}$ if and only if $a(a^p)^* = 0$ in $\mathbb{F}_q[G]$.
\end{proposition}
\begin{proof}
With $M = \sigma(a)$ we have $\overline{M} = \sigma(a^p)$, so $\overline{M}^{\mathsf T} = \sigma(a^p)^{\mathsf T} = \sigma\bigl((a^p)^*\bigr)$. By Proposition~\ref{he-self} and the homomorphism property of $\sigma$,
\[
M\overline{M}^{\mathsf T} = \sigma(a)\,\sigma\bigl((a^p)^*\bigr) = \sigma\bigl(a(a^p)^*\bigr) = O
\iff a(a^p)^* = 0,
\]
the last step by injectivity of $\sigma$.
\end{proof}

We recall the following result of \cite{AK01}, which constructs a quantum code from a Hermitian self-orthogonal code.

\begin{theorem}[\cite{AK01}]\label{H1}
Let $C$ be an $[n, k, d]_{p^2}$ linear code over $\mathbb{F}_{p^2}$. If $C\subseteq C^{\perp_h}$, then there exists a QECC with parameters $[[n, n-2k, \ge d_H]]_p$, where $d_H$ is the minimum Hamming weight of $C^{\perp_h} \setminus C$.
\end{theorem}

\begin{theorem} \label{H2}
Let $G$ be a finite group of order $n$, $q=p^2$, and $a\in\mathbb{F}_q[G]$. Suppose $C=\mathcal{C}(a)$ is a group code with parameters $[n,k,d]_{p^2}$. If $a(a^p)^* = 0$, then there exists a QECC with parameters $[[n, n - 2k, \ge d_H]]_p$.
\end{theorem}
\begin{proof}
This follows from Proposition~\ref{H0} and Theorem~\ref{H1}.
\end{proof}

\begin{example}{\em
Consider the group ring $\mathbb{F}_q[D_m]$ with $q = 3^2$ and $m = 5$. List the dihedral group $D_5$ as
\[
D_5=\{e, b, b^2, b^3, b^4, a, ab, ab^2, ab^3, ab^4\}.
\]
Take $c\in \mathbb{F}_9[D_5]$ given by
\[
c= e + w^5b + 2b^2 + w^6b^3 + w^3b^4 + w^2a + w^6ab + w^7ab^2 + w^7ab^3 + w^7ab^4,
\]
where $w$ is a root of $x^2 + 2x + 2$ in $\mathbb{F}_9$. Then
\[
c^3= e + w^7b + 2b^2 + w^2b^3 + wb^4 + w^6a + w^2ab + w^5ab^2 + w^5ab^3 + w^5ab^4,
\]
\[
(c^3)^*= e + wb + w^2b^2 + 2b^3 + w^7b^4 + w^6a + w^2ab + w^5ab^2 + w^5ab^3 + w^5ab^4.
\]
Let $C$ be the group code generated by $\sigma(c)$. Then $C$ is a $[10, 4, 6]_9$ linear code. A direct calculation gives $c(c^3)^* = 0$, so by Proposition~\ref{H0} the code $C$ is Hermitian self-orthogonal. Theorem~\ref{H2} then yields a QECC with parameters $[[10, 2, 4]]_3$, which match the best known parameters of length $10$ \cite{Gra}.
}\end{example}

\begin{corollary}[Hermitian quantum-MDS family, cyclic case]\label{cor:herm}
Let $q$ be a prime power and $1\le d\le q-1$. From a Hermitian self-orthogonal MDS cyclic code over  $\mathbb{F}_{q^2}$ of length $q^2-1$ and
dimension $d-1$, we obtain the $q$-ary quantum MDS code
\[
  \bigl[\!\bigl[\,q^2-1,\ q^2-2d+1,\ d\,\bigr]\!\bigr]_q .
\]
\end{corollary}

\begin{proof}
Since \(\gcd(q^2-1,q)=1\), \(\mathbb{F}_{q^2}[C_{q^2-1}]\) is semi-simple and every cyclic code of length \(q^2-1\) over \(\mathbb{F}_{q^2}\) is a principal ideal, hence a group code \(C(a)\) for a suitable ambient cyclic group. Let \(\alpha\) be a primitive element of \(\mathbb{F}_{q^2}\) and \(n=q^2-1\); a cyclic code \(C\) of length \(n\) is determined by its defining set
$T\subseteq\mathbb{Z}_n$, with $S:=\mathbb{Z}_n\setminus T$ and
$\dim C=|S|$. Over $\mathbb{F}_{q^2}$ one has
$Z(C^{\perp_H})=\{-q\,z:z\in S\}$, so
\[
  C\subseteq C^{\perp_H}
  \iff (-q\,S)\cap S=\varnothing .
\]
Take $S=\{1,2,\dots,d-1\}$, a run of $d-1\le q-2$ consecutive integers. Then
$T$ is a run of $n-(d-1)$ consecutive integers, so $C$ is a Reed--Solomon code
and is MDS of dimension $d-1$. Moreover, for $1\le j\le d-1$,
\[
  -q\,j \bmod n \;=\; q(q-j)-1 \;\ge\; 2q-1 \;>\; d-1,
\]
so $(-q\,S)\cap S=\varnothing$ and $C$ is Hermitian self-orthogonal. Its dual
is the MDS code $C^{\perp_H}$ with $d(C^{\perp_H})=\dim C+1=d$. Theorem~\ref{H2}
then gives a quantum code of length $n=q^2-1$, dimension
$n-2(d-1)=q^2-2d+1$, and distance $d$, which meets the quantum Singleton bound
and is therefore MDS.
\end{proof}

\begin{remark}{\em
Throughout this section we take $q=p^{2}$ for a prime $p$, since the Hermitian inner product $\langle \mathbf{a},\mathbf{b}\rangle_h$
is defined via the Frobenius automorphism $x\mapsto x^{p}$, which plays the
role of complex conjugation and is available only when the field size is a
perfect square~\cite{AK01,CRSS98}. Under this restriction, Hermitian
self-orthogonal codes over $\mathbb{F}_{p^{2}}$ yield quantum stabilizer codes
over $\mathbb{F}_{p}$; when $q$ is not a perfect square, one instead relies on
Euclidean or symplectic self-orthogonality.
}\end{remark}

\subsection{QECCs from Symplectic Self-Orthogonal Group Codes}

We work with the symplectic inner product on $\mathbb{F}_q^{2n}$. Writing $u=(\alpha_1,\dots,\alpha_n\mid \alpha_{n+1},\dots,\alpha_{2n})$ and $v=(\beta_1,\dots,\beta_n\mid \beta_{n+1},\dots,\beta_{2n})$,
\[
\langle u, v \rangle_s = \sum_{i=1}^{n}\bigl(\alpha_i\beta_{n+i}-\alpha_{n+i}\beta_i\bigr)
= u\,\Omega\, v^{\mathsf T},
\qquad
\Omega=\begin{pmatrix} O_n & I_n \\ -I_n & O_n \end{pmatrix},
\]
where $I_n$ and $O_n$ are the $n\times n$ identity and zero matrices. The symplectic weight of $u$ counts the coordinate pairs that are not jointly zero,
\[
\wt_s(u)=\bigl|\{\,i : (\alpha_i,\alpha_{n+i})\neq(0,0),\ 1\le i\le n\,\}\bigr|,
\]
and the symplectic distance $d_s(C)$ is the smallest symplectic weight among nonzero codewords. The symplectic dual is
\[
C^{\perp_s}=\{\,u\in\mathbb{F}_q^{2n} : \langle u, v\rangle_s=0 \text{ for all } v\in C\,\},
\]
and $C$ is symplectic self-orthogonal when $C\subseteq C^{\perp_s}$.

\begin{remark}\label{rem:sympandgroups}{\em
The symplectic form lives on $\mathbb{F}_q^{2n}$, and there are two natural ways to realize a $2n$-length symplectic code from group rings: take a single $a\in\mathbb{F}_q[G]$ with $|G|=2n$ and use the code generated by $\sigma(a)$; or take two elements $a,b\in\mathbb{F}_q[G]$ with $|G|=n$ and use the horizontally joined matrix $(\sigma(a)\mid\sigma(b))$. We use both below.
}\end{remark}

\begin{remark}{\em
Unlike the Euclidean inner product, which is insensitive to the ordering of the coordinates, the symplectic inner product depends on the partition of the $2n$ coordinates into the two halves of length $n$.
}\end{remark}

We describe two methods for testing symplectic self-orthogonality, corresponding to the two realizations of Remark~\ref{rem:sympandgroups}. In the first (Theorems~\ref{S1} and~\ref{S2}), a single group ring element of a group of order $2n$ generates the code. In the second (Theorems~\ref{S3} and~\ref{S5}), two group codes of length $n$ are joined, so a group ring of order $n$ suffices. We first record the existence result and then the matrix criterion underlying both methods.

\begin{theorem}[\cite{AKS07,AK01}]\label{S0}
Let $C$ be a $[2n,k,d]_q$ linear code over $\mathbb{F}_q$. If $C\subseteq C^{\perp_s}$, then there exists a QECC with parameters $[[n,\,n-k,\,\ge d_S]]_q$, where $d_S$ is the minimum symplectic weight of $C^{\perp_s}\setminus C$.
\end{theorem}

\begin{theorem}[\cite{XD22}]\label{S00}
Let $C$ be a linear code of length $2n$ over $\mathbb{F}_q$ with generating matrix $M$. Then $C\subseteq C^{\perp_s}$ if and only if $M\Omega M^{\mathsf T}=O$.
\end{theorem}

\begin{proof}
The containment $C\subseteq C^{\perp_s}$ holds exactly when every codeword is symplectically orthogonal to every codeword; since the rows of $M$ span $C$, this reduces to orthogonality among the rows, i.e.\ $M\,\Omega\, M^{\mathsf T}=O$.
\end{proof}

\medskip
\noindent\textbf{Symplectic Version 1.}\quad
Specializing $M$ to the generating matrix of a group code gives the following criterion.

\begin{theorem}\label{S1}
Let $G$ be a finite group of order $2n$ and $a\in\mathbb{F}_q[G]$, and let $C=\mathcal{C}(a)$ be the $2n$-length group code with generating matrix $\sigma(a)$. Then $C\subseteq C^{\perp_s}$ if and only if $\sigma(a)\,\Omega\,\sigma(a)^{\mathsf T}=O$.
\end{theorem}

\begin{proof}
Apply Theorem~\ref{S00} with $M=\sigma(a)$.
\end{proof}

\begin{theorem}\label{S2}
Let $G$ be a finite group of order $2n$ and $a\in\mathbb{F}_q[G]$, and suppose $C=\mathcal{C}(a)$ is the $[2n,k,d]_q$ group code with generating matrix $M=\sigma(a)$. If $M\Omega M^{\mathsf T}=O$, then there exists a QECC with parameters $[[n,\,n-k,\,\ge d_S]]_q$, where $d_S$ is the minimum symplectic weight of $C^{\perp_s}\setminus C$.
\end{theorem}

\begin{proof}
Immediate from Theorem~\ref{S1} and Theorem~\ref{S0}.
\end{proof}

\begin{example}{\em
Take the group ring $\mathbb{F}_q[D_m]$ with $q=3$ and $m=11$, and list $D_{11}$ as
\[
D_{11}=\{\,e,b,b^{2},\dots,b^{10},\,a,ab,ab^{2},\dots,ab^{10}\,\}.
\]
Let $c\in\mathbb{F}_3[D_{11}]$ be
\begin{align*}
c &= b^{2}+2b^{4}+b^{5}+2b^{6}+2b^{7}+2b^{8}+2b^{9}+2b^{10}\\
  &\quad +2a+2ab+2ab^{2}+2ab^{3}+2ab^{4}+2ab^{5}+2ab^{6}+2ab^{7}+2ab^{8}+2ab^{9}+2ab^{10},
\end{align*}
whose image under the involution $g\mapsto g^{-1}$ is
\begin{align*}
c^{*} &= 2b+2b^{2}+2b^{3}+2b^{4}+2b^{5}+b^{6}+2b^{7}+b^{9}\\
      &\quad +2a+2ab+2ab^{2}+2ab^{3}+2ab^{4}+2ab^{5}+2ab^{6}+2ab^{7}+2ab^{8}+2ab^{9}+2ab^{10}.
\end{align*}
Let $C$ be the group code generated by $\sigma(c)$. Then $C$ is a $[22,11,6]_3$ linear code, and one checks that $\sigma(c)\,\Omega\,\sigma(c)^{\mathsf T}=O$, so by Theorem~\ref{S1} the code $C$ is symplectic self-orthogonal. Theorem~\ref{S2} then yields a QECC with parameters $[[11,0,5]]_3$, which match the best known parameters at length $11$ \cite{Gra}.
}\end{example}

\medskip
\noindent\textbf{Symplectic Version 2.}\quad
We now realize the $2n$-length code as a horizontal join of two $n$-length group codes. Let $G$ have order $n$, let $a,b\in\mathbb{F}_q[G]$, and set $\mathcal{G}_1=\sigma(a)$, $\mathcal{G}_2=\sigma(b)$. The matrix $\mathcal{G}=(\mathcal{G}_1\mid\mathcal{G}_2)$ generates a $2n$-length linear code, and now $G$ need only have order $n$. Recall $\sigma(a^{*})=\sigma(a)^{\mathsf T}$.

\begin{theorem}\label{S3}
Let $G$ be a finite group of order $n$ and $a,b\in\mathbb{F}_q[G]$, and let $C=C(a,b)$ be the $[2n,k,d]_q$ linear code with generating matrix $\mathcal{G}=(\mathcal{G}_1\mid\mathcal{G}_2)$, where $\mathcal{G}_1=\sigma(a)$ and $\mathcal{G}_2=\sigma(b)$. Then $C\subseteq C^{\perp_s}$ if and only if $ab^{*}=ba^{*}$.
\end{theorem}

\begin{proof}
By Theorem~\ref{S00}, $C\subseteq C^{\perp_s}$ if and only if  $\mathcal{G}\,\Omega\,\mathcal{G}^{\mathsf T}=O$. Computing,
\[
\mathcal{G}\,\Omega\,\mathcal{G}^{\mathsf T}
=\begin{pmatrix}\mathcal{G}_1 & \mathcal{G}_2\end{pmatrix}
 \begin{pmatrix}O_n & I_n\\ -I_n & O_n\end{pmatrix}
 \begin{pmatrix}\mathcal{G}_1^{\mathsf T}\\ \mathcal{G}_2^{\mathsf T}\end{pmatrix}
=\mathcal{G}_1\mathcal{G}_2^{\mathsf T}-\mathcal{G}_2\mathcal{G}_1^{\mathsf T}
=\sigma\!\bigl(ab^{*}-ba^{*}\bigr),
\]
using $\sigma(b)^{\mathsf T}=\sigma(b^*)$ and $\sigma(a)^{\mathsf T}=\sigma(a^*)$. Since $\sigma$ is injective, this is $O$ if and only if  $ab^{*}-ba^{*}=0$.
\end{proof}

\begin{theorem}\label{S5}
With $G$, $a$, $b$, and $C=C(a,b)=[2n,k,d]_q$ as in Theorem~\ref{S3}: if $ab^{*}-ba^{*}=0$, then there exists a QECC with parameters $[[n,\,n-k,\,\ge d_S]]_q$, where $d_S$ is the minimum symplectic weight of $C^{\perp_s}\setminus C$.
\end{theorem}

\begin{proof}
This follows from Theorem~\ref{S3} and Theorem~\ref{S0}.
\end{proof}

\begin{example}{\em
Let $G=\langle g\mid g^{N}=1\rangle$ be the cyclic group of order $N$, and consider the group ring $\mathbb{F}_3 G$. For $a,b\in\mathbb{F}_3 G$ we write $a^{*}$ for the image of $a$ under $g\mapsto g^{-1}$, so $\sigma(a^{*})=\sigma(a)^{\mathsf T}$, and we form the $2N$-length code $C=C(a,b)$ with generating matrix $\mathcal{G}=(\sigma(a)\mid\sigma(b))$. By Theorem~\ref{S3}, $C$ is symplectic self-orthogonal exactly when $ab^{*}-ba^{*}=0$, and then Theorem~\ref{S5} produces a QECC with parameters $[[N,\,N-k,\,\ge d_S]]_3$. The three choices below each improve upon the best known parameters of their length recorded in~\cite{Gra}.

\medskip
\noindent\textbf{(i) A $[[30,3,10]]_3$ code.}\quad Let $N=30$ and
\begin{align*}
a &= 1 + g^2 + g^4,\\
b &= 2 + g + 2g^3 + 2g^6 + 2g^7 + g^9 + g^{10}
      + 2g^{13} + g^{14} + g^{16} + 2g^{17} + 2g^{19} + g^{20} + g^{22}.
\end{align*}
Again $ab^{*}-ba^{*}=0$, so $C=C(a,b)$ is symplectic self-orthogonal. The code $C$ has parameters $[60,27,13]_3$, and Theorem~\ref{S5} gives the quantum code $[[30,3,10]]_3$.

\medskip
\noindent\textbf{(i1) A $[[35,13,8]]_3$ code.}\quad Let $N=35$ and
\begin{align*}
a &= 2 + 2g + g^{3} + 2g^{4} + g^{5} + g^{6} + 2g^{7} + 2g^{8} + 2g^{9}
     + 2g^{10} + g^{11} + 2g^{12} + g^{13},\\
b &= e + g + 2g^{2} + g^{3} + g^{4} + 2g^{5} + 2g^{6} + 2g^{7} + g^{8}
     + 2g^{11} + 2g^{12} + g^{13} + g^{14} + g^{15} + 2g^{16} + g^{17}\\
  &\quad + 2g^{18} + 2g^{19} + g^{20} + g^{21} + 2g^{24} + g^{25} + g^{26}
     + 2g^{27} + 2g^{28} + g^{29} + g^{30} + g^{32} + 2g^{34}.
\end{align*}
Again $ab^{*}-ba^{*}=0$, so $C=C(a,b)$ is symplectic self-orthogonal. The code $C$ has parameters $[70,22,23]_3$, and Theorem~\ref{S5} gives the quantum code $[[35,13,8]]_3$.

\medskip
\noindent\textbf{(iii) A $[[39,16,8]]_3$ code.}\quad Let $N=39$ and
\begin{align*}
a &= e + g + 2g^{2} + 2g^{5} + g^{7} + g^{8} + 2g^{9} + 2g^{10} + 2g^{11}
     + 2g^{12} + g^{13} + g^{14} + 2g^{15} + g^{16},\\
b &= 2g + 2g^{6} + g^{7} + g^{9} + g^{11} + 2g^{12} + g^{13} + 2g^{14}
     + g^{15} + 2g^{16} + 2g^{17} + 2g^{18} + 2g^{20} + g^{21} + 2g^{22}\\
  &\quad + g^{23} + 2g^{27} + 2g^{28} + 2g^{29} + g^{30} + g^{31} + g^{32}
     + g^{33} + 2g^{34} + 2g^{38}.
\end{align*}
A direct computation gives $ab^{*}-ba^{*}=0$, so $C=C(a,b)$ is symplectic self-orthogonal. Here $C$ is a $[78,23,25]_3$ linear code, and Theorem~\ref{S5} yields the quantum code $[[39,16,8]]_3$.

\medskip
\noindent\textbf{(iv) A $[[36,13,8]]_3$ code.} \\
By \cite[Theorem 6]{CRSS98}, if  a quantum code with parameters   $ [[n,k,d]]_p$ exists then a quantum code with parameters $ [[n+1,k,d]]_p$ also exists, when $k > 0$. Therefore, from the above constructed quantum code parameters  $ [[ 35, 13, 8 ]]_3$, we get a quantum code with parameters  $ [[ 36, 13, 8 ]]_3$ which is also new and breaks the current record which is $ [[ 36, 13, 7 ]]_3$; see \cite{Gra}.

\medskip
The codes $[[30,3,10]]_3$, $[[35,13,8]]_3$, $[[ 36, 13, 8 ]]_3$ and $[[39,16,8]]_3$ each exceed the best known parameters of their respective lengths in~\cite{Gra}, showing that the joined construction of Theorem~\ref{S5}, applied over cyclic groups, is an effective source of new quantum codes.
}\end{example}

\section{Conclusion}\label{sec:conclusion}
In this paper we developed group rings as an algebraic framework for constructing both linear and quantum error-correcting codes. Generating the code matrices through Hurley's homomorphism, we treated cyclic, dihedral, direct-product, and semidirect-product groups within a single matrix description, and we extended the analysis to non-abelian structures such as dihedral and quaternion groups. A central observation is that non-isomorphic groups of the same order can generate inequivalent codes: when the attainable dimension sets differ this is forced (Theorem~\ref{thm:distinct}), and even when they coincide the distance profiles can differ, as the computational comparisons of Section~\ref{sec:comparison} show. We used this freedom to obtain larger minimum distances than those available from cyclic codes of the same length.

We derived explicit necessary and sufficient ring-theoretic conditions for group codes to be self-orthogonal under the Euclidean, Hermitian, and symplectic inner products. Combining these conditions with the CSS framework, we constructed quantum codes, and our computational search produced several quantum codes that match or improve upon previously best known parameters, such as the ternary $[[30,3,10]]_3$, $[[35,13,8]]_3$ and $[[39,16,8]]_3$ codes. We also gave a Kronecker-product construction (Theorem~\ref{thm:product}) yielding an infinite family of self-orthogonal group codes and the associated QECCs. Since each fixed-length code admits several generating matrices and we used only a subset of them, a broader search is likely to reveal codes with further improved parameters. Other natural directions include searching over larger non-abelian group libraries, identifying asymptotic families of self-orthogonal group codes, and exploiting the block-circulant structure for efficient quantum decoding.

\section*{Acknowledgment}
Part of this work was carried out while T. Bag was with Inria, ENS de Lyon, and supported by the European Research Council (ERC Grant AlgoQIP, Agreement No.\ 851716) and by the Agence Nationale de la Recherche (ANR) under the France 2030 program, grant ANR-22-PETQ-0006. D. Panario was funded by the Natural Sciences and Engineering Research Council of Canada (NSERC), reference number RPGIN-2024-05341.  The authors thank Prof.\ Markus Grassl for his valuable suggestions, which helped improve this paper, and the reviewers and editor for their constructive feedback.

\end{document}